\documentclass[pra,twocolumn,notitlepage,superscriptaddress,nofootinbib]{revtex4-1}
\usepackage[utf8]{inputenc}
\usepackage[T1]{fontenc}
\usepackage{amsmath}
\usepackage{amssymb}
\usepackage{amsthm}
\usepackage{array}
\usepackage{enumitem}
\usepackage{centernot}
\usepackage{mathtools}
\usepackage{stmaryrd}
\usepackage[usenames,dvipsnames]{xcolor}
\usepackage{float}
\usepackage{pgfplots}
\usepackage{tikz}
\usetikzlibrary{decorations.pathreplacing}
\usepackage[
  bookmarks=false,
  linktocpage, 
  colorlinks=true, 
  citecolor=OliveGreen, 
  urlcolor=blue
]{hyperref}
\usepackage{cleveref}

\crefname{table}{Protocol}{Protocols}

\let\oldaddcontentsline\addcontentsline

\newcommand{\starttocentries}{\let\addcontentsline\oldaddcontentsline}

\newtheoremstyle{definition}
  {3pt}
  {3pt}
  {\sl}
  {}
  {\bf}
  {:}
  {.5em}
  {}
\theoremstyle{definition}
\newtheorem{defn}{Definition}

\newtheorem*{sug}{Suggestion}
\newtheorem{prop}[defn]{Proposition}
\newtheorem*{cor}{Corollary}

\newcommand{\bb}[1]{\mathbb{#1}}
\newcommand{\tr}{\text{tr}}
\newcommand{\posint}{\mathbb{N}_+}
\newcommand{\qtol}{q_{\textnormal{tol}}}
\newcommand{\lstring}[1]{(#1_i)_{i=1}^l}
\newcommand{\mstring}[1]{(#1_i)_{i=1}^m}

\newcommand{\nstring}[1]{(#1_j)_{j=1}^n}

\newcommand{\Mid}{\,\middle|\,}
\newcommand{\abs}[1]{\left\vert#1\right\vert}
\newcommand{\upto}[1]{[#1]}
\newcommand{\expect}[1]{\left\langle #1 \right\rangle}
\newcommand{\lambdatest}{\lambda_{\text{test}}}

\newcommand{\ptail}{p_{\textnormal{tail}}}
\newcommand{\ppass}{p_{\textnormal{pass}}}
\newcommand{\Lambdatest}{\Lambda_{\textnormal{test}}}
\newcommand{\Lambdakey}{\Lambda_{\textnormal{key}}}
\newcommand{\Lambdatot}{\Lambda_{\textnormal{tot}}}
\newcommand{\lcube}{\{0,1\}^l}
\newcommand{\lk}{\{0,1\}^l_k}
\newcommand{\mn}{\{0,1\}^m_n}
\newcommand{\mk}{\{0,1\}^m_k}
\newcommand{\m}{\{0,1\}^m}
\newcommand{\basis}{\vartheta}
\newcommand{\Basis}{\Theta}

\newcommand{\no}[1]{\overline{#1}}
\newcommand{\range}[1]{\mathcal{#1}}
\newcommand{\co}[1]{\text{co}(#1)}
\makeatletter
  \newcommand{\xMapsto}[2][]{\ext@arrow 0599{\Mapstofill@}{#1}{#2}}
  \def\Mapstofill@{\arrowfill@{\Mapstochar\Relbar}\Relbar\Rightarrow}
  \makeatother

\newcommand{\Sigmakey}{\Sigma_{\textnormal{key}}}
\newcommand{\Sigmatest}{\Sigma_{\textnormal{test}}}
\newcommand{\Sigmatot}{\Sigma_{\textnormal{tot}}}
\newcommand{\sigmatot}{\sigma_{\textnormal{tot}}}
\newcommand{\Gammabort}{\Gamma_{\text{abort}}}
\newcommand{\Iabort}{I_{\text{abort}}}
\newcommand{\Ipass}{I_{\text{pass}}}
\newcommand{\pabort}{p_{\text{abort}}}

\DeclareMathOperator{\Forall}{\forall}

\begin{document}

\floatstyle{boxed}
\newfloat{protocol}{htb}{lop}
\floatname{protocol}{Protocol }

\title{Sifting attacks in finite-size quantum key distribution}

\author{Corsin Pfister}
\affiliation{QuTech, Delft University of Technology, Lorentzweg 1, 2628 CJ
  Delft, Netherlands} 
\affiliation{Centre for Quantum Technologies, 3 Science
  Drive 2, 117543 Singapore}
\author{Norbert L\"utkenhaus}
\affiliation{Institute for Quantum Computing and Department of Physics and
  Astronomy, University of Waterloo, N2L3G1 Waterloo, Ontario, Canada}
\author{Stephanie Wehner}
\affiliation{QuTech, Delft University of Technology, Lorentzweg 1, 2628 CJ
  Delft, Netherlands} 
\affiliation{Centre for Quantum Technologies, 3 Science Drive 2, 117543
  Singapore}
\author{Patrick J. Coles}
\affiliation{Institute for Quantum Computing and Department of Physics and
  Astronomy, University of Waterloo, N2L3G1 Waterloo, Ontario, Canada}
\begin{abstract}
  A central assumption in quantum key distribution (QKD) is that Eve has no
  knowledge about which rounds will be used for parameter estimation or key
  distillation. Here we show that this assumption is violated for
  \emph{iterative sifting}, 
  a sifting procedure that has been employed in many (but not all) of the
  recently suggested QKD protocols in order to increase their efficiency.
   
  We show that iterative sifting leads to two security issues: (1) some
  rounds are more likely to be key rounds than others, (2) the public
  communication of past measurement choices changes this bias round by round. We
  analyze these two previously unnoticed problems, present eavesdropping
  strategies that exploit them, and find that the two problems are independent. 

  We discuss some sifting protocols in the literature that are immune to these
  problems. While some of these would be inefficient replacements for iterative
  sifting, we find that the sifting subroutine of an \emph{asymptotically}
  secure protocol suggested by Lo, Chau and Ardehali [\textit{J.\ Cryptol.},
  18(2):133-165, 2005], which we call LCA sifting, has an efficiency on par with
  that of iterative sifting. One of our main results is to show that LCA sifting
  can be adapted to achieve secure sifting in the \emph{finite}-key regime. More
  precisely, we combine LCA sifting with a certain parameter estimation
  protocol, and we prove the finite-key security of this combination. Hence we
  propose that LCA sifting should replace iterative sifting in future QKD
  implementations. More generally, we present two formal criteria for a sifting
  protocol that guarantee its finite-key security. Our criteria may guide the
  design of future protocols and inspire a more rigorous QKD analysis, which has
  neglected sifting-related attacks so far.
\end{abstract}

\maketitle

\section*{Introduction}

Quantum key distribution (QKD) allows for unconditionally secure communication
between two parties (Alice and Bob). A recent breakthrough in the theory of QKD
is the treatment of finite-key scenarios, pioneered by Renner and collaborators
(see~\cite{Renner2005}, for example). This has made QKD theory practically
relevant, since the asymptotic regime associated with infinitely many exchanged
quantum signals is an insufficient description of actual experiments. In
practice, Alice and Bob have limited time, which in turn limits the number of
photons they can exchange. For example, in satellite-based QKD \cite{BTDNV09}
where, say, Bob is on the satellite and Alice is on the ground, the time
allotted for exchanging quantum signals corresponds to the time for the
satellite to pass overhead Alice's laboratory on the ground. Even if such
considerations would not play a role, the necessity of error correction forces
the consideration of finite-size QKD because error correcting codes operate on
blocks of fixed finite length.

Finite-key analysis attempts to rigorously establish the security of finite-size
keys extracted from finite raw data. A systematic framework for such analysis
was developed by Tomamichel et al.~\cite{TLGR12} involving the smooth entropy
formalism. This framework was later extended to a decoy-state protocol by Lim et
al.~\cite{LCW14}. An alternative framework was developed by Hayashi and
collaborators~\cite{Hayashi2011, Hayashi2014}. Other extensions of the
finite-key framework include the treatment of device-independency by Tomamichel
et al.~\cite{TFKW13}, Curty et al.~\cite{CXC14} and Lim et al.~\cite{LPTRG13},
and continuous-variable protocols by Furrer et al.~\cite{Furrer2011} and
Leverrier~\cite{Leverrier2014}.  The framework used in the aforementioned works,
relying on some fairly technical results,\footnote{
  These results include the uncertainty principle for smooth entropies and the
  operational meanings of these entropies.}
represents the current state-of-the-art in the level of mathematical rigor for
QKD security proofs. These theoretical advances have led to experimental
implementations \cite{BCLVV13,Xu2014,Korzh2014} with finite-key analysis.

For practical reasons, it is important to consider not only a protocol's
security but also its efficiency. Ideally a protocol should use as little
quantum communication as possible, for a given length of the final secret key.
For example, it was noted by Lo, Chau and Ardehali \cite{lca05} that---in the
asymptotic regime---protocols with biased basis-choice probabilities can
dramatically decrease the necessary amount of quantum communication per bit of
the raw key.  This is because a bias increases the probability that Alice and
Bob measure in the same basis. As a consequence, when Alice and Bob perform the
\emph{sifting} step of the protocol, where they discard the outcomes of all
measurements that have been made in different bases, they lose less data (see
\Cref{fig:eff-bars} and the discussion in \Cref{sec:fix}).

Some authors have adapted this bias in the basis choice in finite-key protocols
and combined it with another measure to further decrease the amount of data that
is lost through sifting. In the resulting sifting scheme, which we call
\emph{iterative sifting}, Alice and Bob announce previous basis choices while
the quantum communication is still in process, and they
terminate the quantum communication as soon as they have collected sufficiently
many measurement outcomes in identical bases. This way, less quantum
communication takes place, while at the same time they always make sure that
they collect enough data. The implicit assumption here is that the knowledge of
previous basis choices, but not of upcoming ones, does not help a potential
eavesdropper.

As we show in this article, this assumption is wrong. Iterative sifting breaks
the security proofs that have been presented for these protocols. This sifting
scheme was part of theoretical protocols \cite{TLGR12,LPTRG13,CXC14,LCW14} and
has found experimental implementations \cite{BCLVV13}. Therefore, many (but not
all) of the recently suggested protocols in QKD have serious security flaws.

\subsection*{Summary of the results}

The issue with iterative sifting that we point out is as follows.
Typical QKD protocols involve randomly choosing some rounds to be used for
parameter estimation (PE) (i.e.  testing for the presence of an eavesdropper
Eve) and other rounds for key generation (KG). Naturally, if Eve knows ahead of
time whether a round will be used for PE, i.e., if Eve knows which rounds will
form the \emph{sample} for testing for an eavesdropper's presence, then she can
adjust her attack appropriately and the protocol is insecure. Hence a central
assumption in the QKD security analysis is that Eve has no knowledge about the
sample. We show that this assumption is violated for iterative sifting.

To be more precise, the iterative sifting scheme has two problems which, to our
knowledge, have been neither addressed nor noted in the literature: 
\begin{itemize}
  \item \emph{Non-uniform sampling}: The sampling probability, due to which the
    key bits and the encoding basis are chosen, is not uniform.\footnote{
      In general, the sampling probability (which decides over which of the bits
      are chosen as test bits) is distinguished from the probability
      distribution which decides in which basis the information is encrypted. In
      the literature, however, iterative sifting is combined with parameter
      estimation in a way such that bits measured in the $X$-basis are raw key
      bits, and bits measured in the $Z$-basis are used for parameter
      estimation. We will discuss this in more detail in the second half of
      \Cref{sec:itsif}.}
    In other words, there is an a priori bias: Eve knows ahead of time that some
    rounds are more likely to end up in the sample than others.
  \item \emph{Basis information leak}: Alice and Bob's public communication
    about their previous basis choices (which, in iterative sifting, happens
    before the quantum communication is over) allows Eve to update her knowledge
    about which of the upcoming (qu)bits end up in the sample. As a consequence,
    the quantum information that passes the channel thereafter can be correlated
    to this knowledge of Eve. 
\end{itemize}
It is conceivable that these two problems become smaller as the size of the
exchanged data increases. This would remain to be shown. More importantly,
however, the protocols in question are designed to be secure for finite key
lengths. In the light of these two problems, the analysis in the literature does
currently not account for these finite-size effects. This is not a purely
theoretical objection but a practically very relevant issue, as we present some
eavesdropping attacks that exploit the problems.

As we discuss in \Cref{sec:fix}, the basis information leak can trivially be
avoided by fixing the number of rounds in advance, and only announcing the basis
choices after all quantum communication has taken place. We examine some sifting
protocols from the literature with this property. In contrast to protocols that
use iterative sifting, they often use fresh uniform randomness for the choice of
the sample, and therefore are trivially sampling uniformly. This means that they
are secure with respect to our concerns. However, we find that there is room for
improvement over these protocols regarding efficiency aspects.

Concretely, we note that one aspect that makes iterative sifting very efficient
is the parameter estimation protocol that is used with it: after sifting, it
simply uses the $Z$-bits as the sample for parameter estimation and the $X$-bits
for raw key, which is why we call it the \emph{single-basis parameter estimation
{SBPE}}. This is efficient because the sample choice requires no aditional
randomness and no authenticated communication. While SBPE is insecure when used
in conjunction with iterative sifting, it turns out to be secure when used with
a sifting subroutine of a protocol suggested by Lo, Chau and Ardehali, which we
call \emph{LCA sifting}. The combination of LCA sifting and SBPE is essentially
as efficient as iterative sifting. It has trivially no basis information leak
and, as we prove, samples uniformly (see \Cref{prop:uniform-sampling}). We
therefore suggest this combination in future QKD protocols.

More generally, we find clear and explicit mathematical criteria that are
sufficient for a sifting protocol to be secure in combination with SBPE. In
contrast, current literature on QKD does not state such assumptions explicitly,
but rather uses them implicitly.

In our formulation, they take the form of two equations, 
\begin{align}
  &P_{\Basis}(\basis) = P_{\Basis}(\basis') \quad \Forall \basis, \basis' \in
      \lk \quad \text{and} \label{eq:uniform-formal} \\
  &\rho_{A^lB^l\Basis^l} = \rho_{A^lB^l} \otimes \rho_{\Basis^l} \,.
    \label{eq:uncorr-formal}
\end{align}
Here, \Cref{eq:uniform-formal} expresses the absence of non-uniform sampling,
i.e., that the probability $P_{\Basis}(\basis)$ for a partitioning $\basis$ of
the total rounds into sample rounds and key-generation rounds is independent of
$\basis$. \Cref{eq:uncorr-formal} expresses the absence of basis information
leak, which is formally expressed by stating that the classical communication
$\Basis^l$ associated with the sifting process is uncorrelated (i.e., in a
tensor product state) with Alice's and Bob's quantum systems $A^lB^l$. (The
precise details of these two equations will be explained in
\Cref{sec:formal-criteria}.) We find that the two problems are in fact
independent. Hence, security from one of the two problems does not imply
security from the other. The two formal criteria can be used to check whether a
candidate protocol is subject to the two problems or not.

\subsection*{Outline of the paper}

We introduce the iterative sifting protocol in \Cref{sec:itsif}, where we also
explain our conventions and notation. We give a detailed description of the two
problems with iterative sifting in \Cref{sec:problems}.  We show how these
problems can be exploited in \Cref{sec:attacks} by presenting some
intercept-resend attack strategies.

In \Cref{sec:fix}, we discuss some sifting protocols that are immune to these
problems. We study how ideas of existing protocols can be combined to get new
secure protocols that are more efficient. As a result, we suggest the
aforementioned combination of LCA sifting and SBPE, and prove its security.

In \Cref{sec:formal-criteria}, we give a more general answer to the question of
how the two problems can be avoided by presenting formal mathematical criteria
that a sifting protocol needs to satisfy in order to avoid the problems. We
conclude with a summary in \Cref{sec:conclusion}.

\section{Iterative sifting and parameter estimation}
\label{sec:itsif}

A typical QKD protocol consists of the following subroutines \cite{TLGR12}:
\begin{enumerate}[label=(\roman*)]
  \item Preparation, distribution, measurement and sifting, which we
    collectively refer to as ``sifting'', \label{it:sifting}
  \item Parameter estimation, \label{it:paramest}
  \item Error correction, \label{it:errcorr}
  \item Privacy amplification. \label{it:privamp}
\end{enumerate}
What we discuss in this paper refers to the subroutines \ref{it:sifting} and
\ref{it:paramest}, whereas subroutines \ref{it:errcorr} and \ref{it:privamp} are
not of our concern. We refer to subroutine \ref{it:sifting} collectively as
``sifting''. Even though the word sifting usually only refers to the process of
discarding part of the data acquired in the measurements, we refer to the
preparation, distribution, measurement and sifting together as ``sifting'',
because they are intertwined in iterative sifting. 

Our focus in this article is on a particular sifting scheme that we call
iterative sifting. It has been formulated in slightly different ways in the
literature, where the differences lie mostly in the choice of the wording and in
whether it is realized as a prepare-and-measure protocol \cite{TLGR12, BCLVV13,
CXC14, LCW14} or as an entanglement-based protocol \cite{LPTRG13}. These details
are irrelevant for the problems that we describe. Another difference is that
some of the above-mentioned references take into consideration that sometimes, a
measurement may not take place (no-detection event) or may have an inconclusive
outcome. This is done by adding a third symbol $\emptyset$ to the set of
possible outcomes, turning the otherwise dichotomic measurements into
trichotomic ones with symbols $\{0,1,\emptyset\}$. We choose not to do so,
because the problems that we describe arise independently of whether
no-detection events or inconclusive measurements take place. Incorporating them
would not solve the problems that we address but rather complicate things and
distract from the main issues that we want to point out. 

The essence of the iterative sifting protocol is shown in \Cref{prot:iterative}.
There, and in the rest of the paper, we use the notation
\begin{align}
  \upto{r} := \{1, 2, \ldots, r\} \quad \text{for all } r \in \posint \,.
\end{align}
Our formulation of this protocol is close to the one described in \cite{TLGR12},
with the main difference that we choose an entanglement-based protocol instead
of a prepare-and-measure protocol.  This will have the advantage that the formal
criteria in \Cref{sec:formal-criteria} are easier to formulate, but a
prepare-and-measure based protocol would otherwise be equally valid to
demonstrate our points.

\begin{table}
  \textbf{Iterative Sifting} 
  \smallskip
  \hrule
  \begin{description}
    \item[Parameters] $n,k \in \posint$ ; $p_x, p_z \in [0,1]$ with
      $p_x + p_z = 1$.
    \item[Output] For $l=n+k$, the outputs are:
      \begin{itemize}
	\item[Alice:] $l$-bit string $\lstring{s} \in \{0,1\}^l$ (sifted
	outcomes), \item[Bob:] $l$-bit string $\lstring{t} \in \{0,1\}^l$
	  (sifted outcomes), 
	\item[public:] $l$-bit string $\lstring{\basis} \in \{0,1\}^l$ with
	  $\sum_i \basis_i = k$ (basis choices, sifted), where 0 means $X$-basis
	  and 1 means $Z$-basis.
      \end{itemize}
    \item[Number of rounds] Random variable $M$, determined by reaching the
      termination condition (TC) after Step \ref{step:intreport}.
  \end{description}
  \smallskip
  \hrule
  \smallskip
  \textbf{The protocol}
  \begin{description}    
    \item[Loop phase] Steps \ref{step:prep} to \ref{step:intreport} are iterated
      roundwise (round index $r=1,2,\ldots$) until the TC after Step
      \ref{step:intreport} is reached. Starting with round $r=1$, Alice
      and Bob do:
      \begin{enumerate}[label=Step \arabic*:, ref=\arabic*]
	\item \textnormal{(Preparation):} \label{step:prep} Alice prepares a
	  qubit pair in a maximally entangled state.
	\item \textnormal{(Channel use):} \label{step:channel} Alice uses the
	  quantum channel to send half of the qubit pair to Bob.
	\item \textnormal{(Random bit generation):} \label{step:bitgen} Alice
	  and Bob each (independently) generate a random classical bit $a_r$ and
	  $b_r$, respectively, where $0$ is generated with probability $p_x$ and
	  $1$ with probability $p_z$.
	\item \textnormal{(Measurement):} \label{step:meas} Alice measures her
	  share in the $X$-basis (if $a_r = 0$) or in the $Z$-basis (if $a_r =
	  1$), and stores the outcome in a classical bit $y_r$. Likewise, Bob
	  measures his share in the $X$-basis (if $b_r = 0$) or in the
	  $Z'$-basis (if $b_r = 1$), and stores the outcome in a classical bit
	  $y'_r$.
	\item \textnormal{(Interim report):} \label{step:intreport} Alice and
	  Bob communicate their basis choice $a_r$ and $b_r$ over a public
	  authenticated channel. Then they determine the sets
	  \begin{align*} 
	    &u(r) := \{ j \in \upto{r} \mid a_j = b_j = 0 \} \,, \\
	    &v(r) := \{ j \in \upto{r} \mid a_j = b_j = 1 \} 
	  \end{align*}
      \end{enumerate}
      \begin{itemize}
	\item[TC:] If the condition ($\abs{u(r)} \geq n$ and
	  $\abs{v(r)} \geq k$) is reached, Alice and Bob set $m := r$ and
	  proceed with Step \ref{step:discarding}.  Otherwise, they increment
	  $r$ by one and repeat from Step \ref{step:prep}.
      \end{itemize}
    \item[Final phase] The following steps are performed only once:
    \begin{enumerate}[label=Step \arabic*:, ref=\arabic*]
      \setcounter{enumi}{5}
    \item \textnormal{(Random discarding):} \label{step:discarding} Alice and
      Bob choose a subset $u \subseteq u(m)$ of size $n$ at random,
      i.e. each subset of size $k$ is equally likely to be chosen. Analogously,
      they choose a subset $v \subseteq v(m)$ of size $k$ at random.
      Then they discard the bits $a_r$, $b_r$, $y_r$ and $y'_r$ for which $r
      \notin u \cup v$.
    \item \textnormal{(Order-preserving relabeling):} Let $r_i$ be the $i$-th
      element of $u \cup v$. Then Alice determines $\lstring{s} \in
      \{0,1\}^l$, Bob determines $\lstring{t} \in \{0,1\}^l$ and together they
      determine $\lstring{c} \in \{0,1\}^l$, where for every $i \in \upto{l}$,
      \begin{align*}
	s_i = y_{r_i} \,, \quad t_i = y'_{r_i} \,, \quad \basis_i = a_{r_i}
	\left( = b_{r_i} \right) \,.
      \end{align*}
    \item \textnormal{(Output):} \label{step:output} Alice [Bob] locally outputs
      $\lstring{s}$ [$\lstring{t}$], and they publicly output
      $\lstring{\basis}$.
  \end{enumerate}
  \end{description}
  \caption{The iterative sifting protocol. 
    \label{prot:iterative}}
\end{table}

In the protocol, Alice iteratively prepares qubit pairs in a maximally entangled
state (Step \ref{step:prep}) and sends one half of the pair to Bob (Step
\ref{step:channel}).\footnote{
  Choosing a maximally entangled state as the state that Alice prepares
  maximizes the probability that the correlation test in the parameter
  estimation (after sifting) is passed, i.e.  the maximally entangled state
  maximizes the robustness of the protocol.  However, for the security of the
  protocol, which is the concern of the present article, the choice of the state
  that Alice prepares is irrelevant.} 
Then, Alice and Bob each measure their qubit with respect
to a basis $a_i, b_i \in \{0,1\}$, respectively, where $0$ stands for the
$X$-basis and $1$ stands for the $Z$-basis (Steps \ref{step:bitgen} and
\ref{step:meas}). Thereby, Alice and Bob make their basis choice independently,
where for each of them, $0$ ($X$) is chosen with probability $p_x$, and $1$
($Z$) with probability $p_z$. These probabilities $p_x$ and $p_z$ are parameters
of the protocol. The important and problematic parts of the protocol are Step
\ref{step:intreport} and the subsequent check of the termination condition (TC):
after \emph{each} measurement, Alice and Bob communicate their basis choice over
an authenticated classical channel. With this information at hand, they then
check whether the termination condition is satisfied: if for at least $n$ of the
qubit pairs they had so far, they both measured in the $X$-basis, and for at
least $k$ of them, they both measured in the $Z$-basis, the termination
condition is satisfied and they enter the \emph{final phase} of the protocol by
continuing with Step \ref{step:discarding}. These \emph{quota} $n$ and $k$ are
parameters of the protocol. If the condition is not met, they repeat the Steps
\ref{step:prep} to \ref{step:intreport} (which we call the \emph{loop phase} of
the protocol) until they meet this condition.  Because of this iteration, whose
termination condition depends on the history\footnote{
  By the \emph{history} of a protocol run, we mean the record of everything that
  happened during the run of the protocol. In the case of iterative sifting,
  this means the random bits $a_r$, $b_r$, the measurement outcomes $y_r$,
  $y'_r$ etc.}
of the protocol run up to that point,
we call it the iterative sifting protocol. Its number of rounds is a random
variable that we denote by $M$. We denote possible values of $M$ by $m$ (see the
TC and Step~\ref{step:discarding}).

After the loop phase of the protocol, in which the whole data is generated,
Alice and Bob enter the final phase of the protocol, in which this data is
processed. This processing consists of discarding data of rounds in which Alice
and Bob measured in different bases, as well as randomly discarding a surplus of
data for rounds where both measured in the same basis, where a ``surplus''
refers to having more than $n$ ($k$) rounds in which both measured in the $X$
($Z$) basis, respectively. This discarding of surplus is done to simplify the
analysis of the protocol, which is easier if the number of bits where both
measured in the $X$ ($Z$) basis is fixed to a number $n$ ($k$). Since after the
loop phase, Alice and Bob can end up with more bits measured in this same basis,
they throw away surplus at random. Finally, after throwing away the surplus,
Alice and Bob locally output the remaining bit strings $\lstring{s}$ and
$\lstring{t}$ of measurement outcomes and publicly output the remaining bit
string $\lstring{\basis}$ of basis choices. 

Iterative sifting is problematic, but to fully understand why, one needs to see
how the output of the iterative sifting protocol is processed in the subsequent
subroutine \ref{it:paramest}, the parameter estimation, where Alice and Bob
check for the presence of an eavesdropper. 
Protocols that use iterative sifting use a particular protocol for parameter
estimation. To make clear what we are talking about, we have written it out in
\Cref{prot:paramest}.

Alice and
Bob start the protocol with the strings $\lstring{s}$, $\lstring{t}$ and
$\lstring{\basis}$ that they got from sifting.  Then, in a first step, they
communicate the \emph{test bits}. The test bits are those bits $s_i$, $t_i$ that
resulted from measurements in the $Z$-basis, i.e.  the bits $s_i$, $t_i$ with
$i$ such that $\basis_i=1$. Then, they determine the fraction of the test bits
that are different for Alice and Bob, i.e. they determine the \emph{test bit
error rate}. If it is higher than a certain protocol parameter $\qtol \in
[0,1]$, they abort.  Otherwise, they locally output the \emph{raw keys}, which
are the bits $s_i$, $t_i$ that result from measurements in the $X$-basis, i.e.
those $s_i$, $t_i$ with $i$ for which $\basis_i=0$.

\begin{table}
  \textbf{Single-Basis Parameter Estimation (SBPE)} 
  \smallskip
  \hrule
  \begin{description}
    \item[Protocol Parameters] $n,k \in \posint$, $p_x, p_z \in [0,1]$ with
      $p_x+p_z=1$ and $\qtol \in [0,1]$.
    \item[Input] For $l=n+k$, the inputs are:
      \begin{itemize}
	\item[Alice:] $l$-bit string $\lstring{s} \in \{0,1\}^l$ (measurement
	  outcomes, sifted),
	\item[Bob:] $l$-bit string $\lstring{t} \in \{0,1\}^l$ (measurement
	  outcomes, sifted), 
	\item[public:] $l$-bit string $\lstring{\basis} \in \{0,1\}^l$ with
	  $\sum_i \basis_i = k$ (basis choices, sifted), where 0 means $X$-basis
	  and 1 means $Z$-basis.
      \end{itemize}
    \item[Output] Either no output (if the protocol aborts in Step
      \ref{step:corrtest}) or:
      \begin{itemize}
	\item[Alice:] $n$-bit string $\nstring{x} \in \{0,1\}^n$ (raw key),
	\item[Bob:] $n$-bit string $\nstring{x'} \in \{0,1\}^n$ (raw key). 
      \end{itemize}
  \end{description}
  \smallskip
  \hrule
  \smallskip
  \textbf{The protocol}
  \begin{enumerate}[leftmargin=1.3cm, label=Step \arabic*:, ref=\arabic*] 
    \item \textnormal{(Test bit communication):} Alice and Bob communicate their
      test bits, i.e. the bits $s_i$ and $t_i$ with $i$ for which $\basis_i=1$,
      over a public authenticated channel.
    \item \textnormal{(Correlation test):} \label{step:corrtest} Alice and Bob
      determine the \emph{test bit error rate}
      \begin{align*}
	\lambdatest := \frac{1}{k} \sum_{i=1}^l \basis_i (s_i \oplus t_i) \,,
      \end{align*}
      where $\oplus$ denotes addition modulo 2, and do the \emph{correlation
      test}: if $\lambdatest \leq \qtol$, they continue the protocol and move on
      to Step \ref{step:raw-key-output}.  If $\lambdatest > \qtol$, they abort.
    \item \textnormal{(Raw key output):} \label{step:raw-key-output} Let $i_j$
      be the $j$-th element of $\{i\in\upto{l} \mid \basis_i=0\}$. Then Alice
      outputs the $n$-bit string $\nstring{x}$ and Bob outputs the $n$-bit
      string $\nstring{x'}$, where
      \begin{align*}
	x_j = s_{i_j} \,, \quad x'_j = t_{i_j} \,.
      \end{align*}
  \end{enumerate}
  \caption{The single-basis parameter estimation (SBPE) protocol. 
    \label{prot:paramest}}
\end{table}

It is important to emphasize that if the output of iterative sifting serves as
the input of the parameter estimation protocol as in \Cref{prot:paramest}, then
the bits that result from measurements in the $X$-basis are used for the raw
key, and the bits that result from measurements in the $Z$-basis are used for
parameter estimation (i.e. they form the \emph{sample} for the parameter
estimation). Hence, the sample is determined by the basis choice; no additional
randomness is injected to choose the sample. 
  We call this the \emph{single-basis parameter estimation (SBPE)}, because the
  parameter estimation is done in only one basis.

This is not necessarily a problem by itself. However, as we will show in
\Cref{sec:non-uniform-sampling}, in iterative sifting, some rounds are more
likely to end up in the sample than other rounds.  This leads to non-uniform
sampling, which is a problem since uniform sampling is one of the assumptions
that enter the analysis of the parameter estimation. This seems to be unnoticed
so far, as we found that protocols in the literature that use iterative sifting
as a subroutine use SBPE as a subroutine for parameter estimation (or
something equivalent) \cite{TLGR12, LPTRG13, CXC14, LCW14, BCLVV13}. In
contrast, the LCA sifting protocol that we discuss in
\Cref{sec:fix} \emph{does} sample uniformly, even if bits from
$X$-measurements are used for the raw key and $Z$-measurements are used for
paremeter estimation, without injecting additional randomness.

We will discuss randomness injection for the sample choice in more detail in
\Cref{sec:fix}.

The idea behind the parameter estimation is the following: if the correlation
test passes, then the likelihood that Eve knows much about the raw key is
sufficiently low. The exact statement of this is subtle, and involves more
details than are necessary for our purposes. We refer to \cite{TLGR12} for more
details.  Here, what is important is that this estimate of Eve's knowledge is
done via estimating another probability that we call the \emph{tail probability}
$\ptail(\mu)$ which, for $\mu \in [0,1]$, is given by
\begin{align}
  \label{eq:ptail-first}
  \ptail(\mu) = P[\Lambdakey \geq \Lambdatest + \mu \mid \Lambdatest \leq \qtol]
  \,. 
\end{align}
Here, $\Lambdatest$ is the random variable of the test bit error rate
$\lambdatest$ determined in the parameter estimation protocol, 
  \begin{align}
    \lambdatest := \frac{1}{k} \sum_{i=1}^l \basis_i (s_i \oplus t_i) \,.
  \end{align}
The random variable $\Lambdakey$ is the random variable of a quantity that is
not actually measured: it is the random variable of the error rate on the raw
key bits \emph{if they had been measured in the $Z$-basis}. Since in the actual
protocol, the raw key bits have been measured in the $X$-basis, the random
variable $\Lambdakey$ is the result of a \emph{Gedankenexperiment} rather than
an actually measured quantity. We will define $\Lambdakey$ formally in
\Cref{sec:formal-criteria}. 

The usual analysis, as in Reference \cite{TLGR12}, aims at proving that
\begin{align}
  \label{eq:ptail-bound}
  \ptail(\mu) &\leq \frac{\exp\left( -2\frac{kn}{l}\frac{k}{k+1}\mu^2
  \right)}{\ppass} \,, 
\end{align}
where
\begin{align}
  \label{eq:ppass}
  \ppass = P[\Lambdatest \leq \qtol]
\end{align}
Inequality \eqref{eq:ptail-bound} is turned into an inequality about the
eavesdropper's knowledge about the raw key using an uncertainty relation for
smooth entropies \cite{TR11,TLGR12}.

\subsection*{Notation and terminology}
\label{sec:notation}

In the following sections, we will have a closer look at the probabilities of
certain outputs of the iterative sifting protocol in \Cref{prot:iterative}. For
example, in \Cref{sec:non-uniform-sampling} we will consider the probability
that iterative sifting with parameters $n=1$, $k=2$ outputs the string $\basis =
(\basis_i)_{i=1}^3 = (1,1,0)$. Since the output of the protocol is
probabilistic, the output string becomes a random variable. We denote random
variables by capital letters and their values by lower case letters. For
example, the random variable for the output string $\basis$ is denoted by
$\Basis$, and the probability of the output string to have a certain value
$\basis$ is $P[\Basis=\basis]$. For strings in
$\basis=\lstring{\basis}\in\{0,1\}^l$, we write $\lstring{\basis} =
\basis_1\basis_2\ldots\basis_l$ instead of $\lstring{\basis} =
(\basis_1,\basis_2,\ldots,\basis_l)$, i.e. we omit the brackets and commas. For
example, we write $110 \in \{0,1\}^3$ instead of $(1,1,0) \in \{0,1\}^3$, so the
probability that we calculate in \Cref{sec:non-uniform-sampling} is
$P[\Basis=110]$. Other random variables that we consider include the random
variable $A_1$ ($B_1$) of Alice's (Bob's) first basis choice $a_1$ ($b_1$) or
the random variable $M$ of the number $m$ of total rounds performed in the loop
phase of the iterative sifting protocol.

To simplify the calculations, it is convenient to introduce the following
terminology. For a round $r$ in the loop phase of the iterative sifting
protocol, $r$ is an $X$-agreement if $a_r = b_r = 0$, $r$ is a $Z$-agreement if
$a_r = b_r = 1$ and $r$ is a disagreement if $a_r \neq b_r$. We sometimes say
that $r$ is an agreement if it is an $X$- or a $Z$-agreement.

For calculations with random variables like $\Basis$, $A_1$, $B_1$ or $M$, the
sample space of the relevant underlying probability space is the set of all
possible histories of the iterative sifting protocol. This set is hard to model,
as it contains not only all possible strings $(a_r)_r$, $(b_r)_r$, $(y_r)_r$ and
$(y'_r)_r$ of the loop phase (which can be arbitrarily long) but also a record
of the choice of the subsets $u$ and $v$ in the random discarding
during the final phase. It is, however, not necessary for our calculations to
have the underlying sample space explicitly written out. In order to avoid
unnecessarily complicating things, we therefore only deal with the relevant
events, random variables and their probability mass functions directly, assuming
that the reader understands what probability space they are meant to be defined
on. In contrast, the LCA sifting protocol which we discuss in
\Cref{sec:fix}, has a simpler set of histories, and we will derive a probability
space model for it in \Cref{app:fixed-round}.

We often write expressions in terms of probability mass functions instead of in
terms of probability weights of events, e.g. we write
\begin{align}
  P_{\Basis}(\basis) := P[\Basis =\basis] \,. 
\end{align}

\section{The problems}
\label{sec:problems}

\subsection{Non-uniform sampling}
\label{sec:non-uniform-sampling}

To show that iterative sifting leads to non-uniform sampling, we calculate the
sampling probabilities for some example parameters $k,n \in \posint$ as
functions of the probabilities $p_x$ and $p_z$. By a sampling probability, we
mean the probability that some subset of $k$ of the $l=n+k$ bits is used as a
sample for the parameter estimation, i.e. the sampling probabilities are
$P_\Basis(\basis)$ for $\basis \in \lk$, where
\begin{align}
  \lk := \left\{ \lstring{\basis} \in \lcube \middle\vert \sum_{i=1}^l \basis_i = k
  \right\}
\end{align}
is the set of all $l$-bit strings with Hamming weight $k$. We say that sampling
is uniform if $P_\Basis(\basis)$ is the same for all $\basis \in \lk$, and
non-uniform otherwise. While non-uniform sampling already arises in the case of
the smallest possible parameters $k=n=1$, the results are even more interesting
in cases where $k \neq n$. Let us consider iterative sifting
(\Cref{prot:iterative}) with $n=1$, $k=2$ and arbitrary $p_x, p_z \in [0,1]$.
Let $\Basis $ denote the random variable of the string $\basis =
(\basis_i)_{i=1}^3 = \basis_1\basis_2\basis_3$ of sifted basis choices which is
generated by the protocol. The possible values of $\Basis $ are $110$, $101$ and
$011$. The probabilities of these strings are given as follows (see
\Cref{app:sample-prob} for a proof).

\begin{prop}
  \label{prop:sample-prob}
  For the iterative sifting protocol as in \Cref{prot:iterative} with $n=1$ and
  $k=2$, it holds that
  \begin{align}
    P_\Basis(110) = g_z^2 \,, \quad \text{where} \quad g_z =
    \frac{p_z^2}{p_z^2+p_x^2} \,. 
  \end{align}
  For the other two possible values of $\Basis $, it holds that
\begin{align}
  P_\Basis (011) = P_\Basis (101) = \frac{1-g_z^2}{2} \,.
  \label{eq:other-partitions2}
\end{align}
\end{prop}

Hence, different samples have different probabilities, in general. In order for
the sampling probability $P_\Basis$ to be uniform, in the case where $n=1$ and
$k=2$, we need to have $P_\Basis (\basis) = 1/3$ for $\basis=011,101,110$. This
holds if and only if $g_z = g_z^*$, where $g_z^* = 1/\sqrt{3}$, which in turn is
equivalent to $p_z = p_z^*$, where
\begin{align}
  \label{eq:equalizing-p_z-2}
  p_z^* = \frac{\left(3+2\sqrt{3}\right)\left( 1+\sqrt{\sqrt{3}-1}
  \right)}{\sqrt{3}} \approx 0.539 \,.
\end{align}
This is bad news for iterative sifting: it means that iterative sifting leads to
non-uniform sampling for all values of $p_z$ except $p_z=p_z^*$. Interestingly,
the value of $p_z^*$ does not seem to be a probability that has been considered
in the QKD literature. In particular, $p_z^*$ corresponds to neither the
symmetric case $p_z=1/2$ nor to a certain asymmetric probability which has been
suggested to be chosen in order to maximize the key rate \cite{TLGR12}.

The value $g_z$ can be interpreted as the probability that in a certain round of
the loop phase, Alice and Bob have a $Z$-agreement, given that they have an
agreement in that round (this conditional is why the $p_z^2$ is renormalized
with the factor $1/(p_z^2+p_x^2)$). Hence, $g_z^2$ is the probability that Alice
and Bob's first two basis agreements are $Z$-agreements. Therefore, $P_\Basis
(110) = g_z^2$ is what one would intuitively expect: to end up with $\Basis
=110$, the first two basis agreements need to be $Z$-agreements, and conversely,
whenever the first two basis agreements are $Z$-agreements, Alice and Bob end up
with $\Basis =110$.

More generally, it turns out that for $n=1$ and for $k \in \posint$ arbitrary,
the iterative sifting protocol leads to
\begin{align}
  P_\Basis (1\ldots10) &= g_z^k \,, \\
   P_\Basis (\basis) &= \frac{1-g_z^k}{k}
  \quad \text{for all other } \basis \in \lk \,.
\end{align}
This is a uniform probability distribution if and only if $g_z = g_z^*$, where
\begin{align}
  \label{eq:equalizing-g}
  g_z^* = \left( \frac{1}{k+1} \right)^{1/k} \,,
\end{align}
which is true iff $p_z = p_z^*$, where
\begin{align}
  \label{eq:equalizing-p}
  p_z^* = \frac{g_z^*-\sqrt{g_z^*(1-g_z^*)}}{2g_z^*-1} \,.
\end{align}

Hence, we conclude that iterative sifting does not lead to uniformly random
sampling, unless $p_x$ and $p_z$ are chosen in a very particular way. This
particular choice does not seem to correspond to anything that has been
considered in the literature so far. 

\subsection{Basis information leak}

In iterative sifting, information about Alice's and Bob's basis choices reaches
Eve in every round of the loop phase. In Step \ref{step:intreport} of round $r$,
Alice and Bob communicate their basis choice $a_r$, $b_r$ of that round. They do
so because they want to condition their upcoming action on the strings
$a_1\ldots a_r$ and $b_1\ldots b_r$: if they have enough basis agreements, they
quit the loop phase; otherwise they keep looping. 

What seems to have remained unnoticed in the literature is that Eve can also
condition her actions on $a_1\ldots a_r$ and $b_1\ldots b_r$. This means that if
there is a round $r+1$, Eve can correlate the state of the qubit that Alice
sends to Bob in round $r+1$ with $a_1\ldots a_r$ and $b_1\ldots b_r$. Hence, the
state of the qubit that Bob measures is correlated with the classical register
that keeps the information about the basis choice. Note that the basis
information leak tells Eve how close Alice and Bob are to meeting their quotas
for each basis.  Eve can tailor her attack on future rounds based on this
information.  For example, if Alice and Bob have already met their $Z$-quota,
but not their $X$-quota, then Eve can measure in the $X$-basis, knowing that, if
Alice and Bob happen to both measure $Z$, the round may be discarded anyway. 

We want to emphasize that the basis information leak is not resolved by
injecting additional randomness for the choice of the sample. As we will discuss
in \Cref{sec:fix}, such additional randomness can ensure that the sampling is
uniform, but it does not help against the basis information leak. Randomness
injection for the sample is effectively equivalent to performing a random
permutation on the qubits \cite{ren07}. This does not remove the correlation
between the classical basis information register and the qubits.

We will see more concretely how the basis information leak is a problem when we
present an eavesdropping attack in \Cref{sec:attack-on-non-uniform} and when we
treat the problem more formally in \Cref{sec:formal-criteria}.

\section{Eavesdropping attacks}
\label{sec:attacks}

A detailed analysis of the effect of non-uniform sampling and basis information
leak on the key rate is beyond the scope of the present paper. It would involve
developing a new security analysis for a whole protocol involving iterative
sifting. Instead of attempting to find a modified analysis for iterative
sifting, we will discuss alternative protocols in \Cref{sec:fix}.

However, to give an intuitive idea of the effect, we will calculate another
figure of merit: the error rate for an intercept-resend attack. We devise a
strategy for Eve to attack the iterative sifting protocol during its loop phase
and calculate the expected value of the error rate 
\begin{align}
  \label{eq:error-rate}
  E = \frac{1}{l} \sum_{i=1}^l S_i \oplus T_i
\end{align}
that results from this attack. Here, $\oplus$ denotes addition modulo 2 and
$S_i$ and $T_i$ are the random variables of the bits $s_i$ and $t_i$,
respectively, which are generated by the protocol. One would typically expect an
error rate no lower than $25\%$ for an intercept-resend attack \cite{he94},
which is why our results below are alarming.

\subsection{Attack on non-uniform sampling}
\label{sec:attack-on-non-uniform}

Let us first consider an attack on non-uniform sampling, i.e., on the fact that
not every possible value of $\Basis$ is equally likely. It will be a particular
kind of intercept-resend attack, i.e. Eve intercepts all the qubits that Alice
sends to Bob during the loop phase, measures them in some basis and afterwards,
prepares another qubit in the eigenstate associated with her outcome and sends
it to Bob.  Then we will show that the attack strategy leads to an error rate
below $25\%$.

For the error rate calculation, we assume that the $X$- and $Z$-basis is the
same for Alice, Bob and Eve, and that they are mutually unbiased. This way, if
Alice and Bob measure in the same basis, but Eve measures in the other basis,
then Eve introduces an error probability of $1/2$ on this qubit. Moreover, for
simplicity, we make this calculation for the easiest possible choice of
parameters. Consider the iterative sifting iterative sifting protocol
(\Cref{prot:iterative}) with the parameters $k=n=1$. From
Equations~\eqref{eq:equalizing-g} and \eqref{eq:equalizing-p}, we get that the
sampling probabilities in this case are
\begin{align}
  \label{eq:k=n=1-sampling-probs}
  P_\Basis (01) = \frac{p_x^2}{p_x^2 + p_z^2} \,, \quad P_\Basis (10) =
  \frac{p_z^2}{p_x^2 + p_z^2} \,.
\end{align}
These sampling probabilities are uniform for the symmetric case $p_x = p_z$, but
are non-uniform for all other values. In the following, we assume $p_x>1/2$,
which makes the sample $\Basis =01$ more likely than the sample $\Basis =10$. We
choose the following attack: in the first round of the loop phase, she attacks
in the $X$-basis, and in all the other rounds, she attacks in the $Z$-basis. We
choose the attack this way because we know that the first non-discarded
basis agreement is more likely to be an $X$-agreement, whereas the second one is
more likely to be a $Z$-agreement.\footnote{
  The attentive reader may point out that this attack could be improved by
  making Eve's basis choice dependent on the communication between Alice and
  Bob. This is correct, but we intentionally design the attack such that Eve
  ignores Alice and Bob's communication. That allows one to see the effect of
  non-uniform sampling alone and to compare it to attacks on basis
  information leak alone, see \Cref{sec:attack-on-basis,sec:independence}.}

We calculate the expected error rate for this attack in \Cref{app:attack1}. The
black curve in \Cref{fig:attack1} shows $\expect{E}$ as a function of
$p_x$ for this attack.  Notice that $\expect{E}$ falls below $25\%$ for $1/2 <
p_x < 1$, and reaches a minimum of $\expect{E} \approx 22.8\%$ for $p_x \approx
0.73$. 

\begin{figure}[tbp]
  \centering
  \begin{tikzpicture}
    \node[anchor=south west,inner sep=0] at (0,0)
    {\includegraphics[width=8cm]{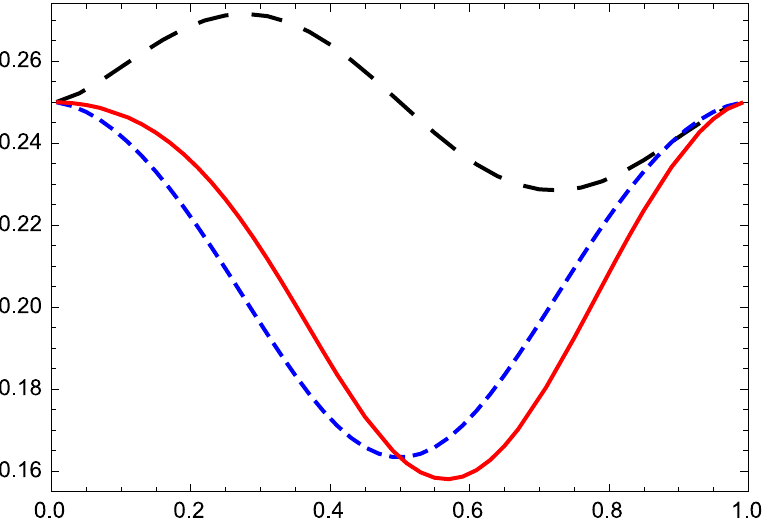}}; 
    \node at (4.2,-0.2) {$p_x$};
    \node at (-0.3,2.7) {$\expect{E}$};
  \end{tikzpicture}
  \caption{The error rate for three different eavesdropping attacks
    iterative sifting: (1) attack
    on non-uniform sampling (long-dashed, black curve), (2) attack on
    basis-information leak (short-dashed, blue curve), (3) attack on both
    problems (solid, red curve).
    \label{fig:attack1}}
\end{figure}

The concerned reader might worry that the $25\%$ error rate associated with the
intercept-resend attack was derived under the assumption of equal weighting for
the two bases $X$ and $Z$, whereas it seems here that we choose unequal
weightings.  However, for the protocol under consideration, the a priori
probability distribution $\{p_x, p_z\}$ is not the relevant quantity. Rather,
the fact that $n = k$ in our example ensures that the X and Z bases enter in
with equal weighting.

\subsection{Attack on basis information leak}
\label{sec:attack-on-basis}

We now give an eavesdropping strategy that exploits the basis information leak.
It is an adaptive strategy, in which Eve's action in round $r+1$ depend on the
past communication of the strings $a_1\ldots a_r$ and $b_1 \ldots b_r$. Again,
we consider the simple case of $n=k=1$. To make sure our attack is really
exploiting the basis information leak and not the non-uniform sampling, we set
$p_x=p_z=1/2$. In this case, from Eq.~\eqref{eq:k=n=1-sampling-probs}, the
sampling is uniform:
\begin{align}
  P_\Basis (01) = P_\Basis (10) = \frac{1}{2}.
\end{align}

Before we define Eve's strategy, we want to give some intuition. Suppose that
during the protocol, Eve learns that Alice and Bob just had their first basis
agreement. If this first agreement is a $Z$-agreement, say, what does this mean
for Eve? She knows that the protocol will now remain in the loop phase until
they end up with an $X$-agreement. Suppose that she now decides that she will
measure all the remaining qubits in the $X$-basis. Then, if the next basis
agreement of Alice and Bob is an $X$-agreement, Eve knows the raw key bit
perfectly, and her measurement on that bit did not introduce an error. If the
next basis agreement is a $Z$-agreement, she may introduce an error on that test
bit. However, there will be a chance that Alice and Bob discard this test bit,
because they have a total of two (or more, in the end) $Z$-agreements, and the
protocol forces them to discard all $Z$-agreements except $k=1$ of them. Hence,
learning that the first basis agreement was a $Z$-agreement brings Eve into an
favorable position: she knows that attacking in the $X$-basis for the rest of
the loop phase will necessarily tell her the raw key bit, while she has quite
some chance to remain undetected. 

This intuition inspires the following intercept-resend attack. Before the first
round of the loop phase, Eve flips a fair coin. Let $F$ be the random variable
of the coin flip outcome and let $0$ and $1$ be its possible values. If $F=0$,
then in the first round, Eve attacks in the $X$ basis, and if $F=1$, she attacks
in the $Z$-basis.  In the subsequent rounds, she keeps attacking in that basis
until Alice and Bob first reached a basis agreement. If it is an $X$-agreement
(equivalent to $\Basis =01$), Eve attacks in the $Z$-basis in all remaining
rounds, and if it is a $Z$-agreement (equivalent to $\Basis =10$), she attacks
in the $X$-basis in all remaining rounds.\footnote{
  We let Eve flip a coin in order to make the attack symmetric between $X$ and
  $Z$. This allows for a more meaningful comparison with the attack on
  non-uniform sampling, as this attack here does not exploit non-uniform
  sampling even if $p_x \neq 1/2$, see \Cref{sec:attack-on-non-uniform,sec:independence}.}

We calculate the expected error rate for this attack in the \Cref{app:attack2}.
We find that
\begin{align}
  \expect{E} = \frac{2-\ln2}{8} \approx 16.3\% \,.
\end{align}
Hence, the basis information leak allows Eve to go far below the typical
expected error rate of $25\%$ for intercept-resend attacks \cite{SBP09}. The
blue curve in \Cref{fig:attack1} shows, more generally,
$\expect{E}$ as a function of $p_x$, for this attack.

\subsection{Independence of the two problems}
\label{sec:independence}

Are non-uniform sampling and basis information leak really two different
problems, or is one a consequence of the other? We will argue now that the two
problems are in fact independent. To this end, we describe a protocol that
suffers from non-uniform sampling but not from basis information leak, and
another protocol that suffers from basis information leak but not from
non-uniform sampling.

We have already seen an instance of a protocol that suffers from basis
information leak but not from non-uniform sampling: in
\Cref{sec:attack-on-basis}, we looked at the iterative sifting protocol with
$n=k=1$ and $p_x=p_z=1$, in which case the sampling is uniform. Hence, there was
no exploitation of non-uniform sampling, but the attack strategy exploited basis
information leak.

What about the other way round? Can non-uniform sampling occur without
basis information leak? A closer look at the attack on non-uniform sampling
presented in \Cref{sec:attack-on-non-uniform} hints that this is possible: the
attack strategy works, even though it completely ignores the communication
between Alice and Bob, so it did not make any use of the basis information leak
due to this communication.

A more dramatic example shows clearly that non-uniform sampling can occur
without basis information leak. To this end, we forget about iterative sifting
for a moment and look at a different protocol. Consider a sifting-protocol in
which Alice and Bob agree in advance that they will measure the first $n=100$
qubits in the $X$-basis, and that they will measure the second $k=100$ qubits in
the $Z$-basis, without any communication during the protocol. Of course, there
is no hope for this protocol to be useful for QKD, but it serves well to
demonstrate our point. It leads to a very dramatic form of non-uniform sampling,
because $P_\Basis (0\ldots01\ldots1)=1$ and $P_\Basis (\basis)=0$ for all other
$\basis\in\lk$. If Eve attacks the first $100$ rounds in $X$ and the second
$100$ rounds in $Z$, then she knows the raw key perfectly, without introducing
any error. At the same time, there is no communication between Alice and Bob
during the protocol, so no information about the basis choice is \emph{leaked
during the protocol}. Instead, Eve (who is always assumed to know the protocol)
already had this information before the first round.

Hence, we conclude that the problems of non-uniform sampling and basis
information leak are independent. They just happen to occur simultaneously for
iterative sifting, but they can occur separately in general.  We will see the
independence of the two problems more formally in \Cref{sec:formal-criteria}.

\subsection{Attack on both problems}
\label{sec:attack-on-both}

Since the two problems are independent, it is interesting to devise an attack
that exploits both of them. Let us again consider $k=n=1$ and suppose $p_x>1/2$
to ensure that we have non-uniform sampling. Suppose Eve begins in the same way
as in the attack on non-uniform sampling, measuring in the $X$-basis. However,
as in the attack on the basis-information leak, she makes her attack adaptive by
following the rule that she switches to the $Z$-basis when Alice and Bob
announce that they had an $X$-agreement. If Alice and Bob announce a
$Z$-agreement, Eve keeps attacking in the $X$-basis. 

We give an expression for the error rate induced by this attack in
\Cref{app:attack3}. The red curve in \Cref{fig:attack1} shows a plot of
this error rate as a function of $p_x$. As one can see, the error rate attains
its minimum of $\expect{E} \approx 15.8\%$ for $p_x \approx 0.57$. Hence, this
combined attack on both problems performs much better than the one on
non-uniform sampling alone (with a minimal error rate of $\sim 22.8\%$) and
even better than the attack on the basis information leak alone (with a minimal
error rate of \mbox{$\sim16.3\%$}).

\section{Solutions to the problems}
\label{sec:fix}

How can these problems be avoided? Roughly speaking, we can say that protocols
with iterative sifting are characterized by three properties that make it
efficient: (1) asymmetric basis choice probabilities and quota, $p_x > p_z$ and
$n > k$, (2) single-basis parameter estimation (\Cref{prot:paramest}), (3)
communication in Step \ref{step:intreport} of the loop phase. As we have seen,
it is the communication which causes the basis information leak.

An obvious fix to this problem is to take this communication out of the loop
phase and to postpone it to the final phase, when all the quantum communication
is over. Then there is no classical communication during the loop phase, and
hence, there cannot be a termination condition that depends on classical
communication.  Instead, the number of rounds in the loop phase is set to a
fixed number $m \in \posint$. This number $m$ then becomes a parameter of the
protocol.

Fixing the number of rounds introduces a new issue: there is no guarantee that
the quotas for $X$- and $Z$- agreements will be met after $m$ rounds. In order
to perform the parameter estimation, however, the quotas $n$ and $k$ must be
met. Otherwise, Inequality \eqref{eq:ptail-bound} is not applicable, because the
number of $X$- and $Z$-agreements in the loop phase are random numbers that can
be below $n$ and $k$, respectively. Thus, unless one wants to introduce a new
tail probability analysis as well, there is a strictly positive probability that
Alice and Bob have to abort the sifting protocol because they have too many
basis disagreements. If the sifting scheme is modified in this way, it no longer
involves any communication about the basis choices during its loop phase. Thus,
it is trivially true that there is no basis information leak.

Many protocols in the QKD literature have such a fixed number $m$ of rounds
(which is often denoted by $N$ instead) and an according abort event. It seems
that before iterative sifting was introduced, the sifting procedure was either
not clearly written out in the protocols, or it had such a fixed round number.
For example, in the original BB84 paper \cite{BB84}, the sifting scheme is not
written out in enough detail to say whether this is the case, but the protocol
for which Shor and Preskill showed asymptotic security uses a fixed number of
rounds \cite{SP00}. In addition, they use symmetric basis choice probabilities
and quota, i.e. $p_x=p_z=1/2$ and $k=n$. Alice sends $4n+\delta$ qubits to Bob
(where $\delta$ is a positive but small overhead) without any intermediate
classical communication. Afterwards, they compare their bases and check whether
they have at least $n$ $X$-agreements and at least $n$ $Z$-agreements. If not,
they abort, otherwise they choose $n$ $X$-agreements and $n$ $Z$-agreements and
discard the rest. 

With the remainin $2n$ bits, they continue with parameter estimation. However,
instead of performing SBPE, they choose $n$ bits at random (i.e. with fresh
randomness) for parameter estimation and use the rest for the raw key. Hence,
this protocol shares none of the three properties with iterative sifting that we
listed above.

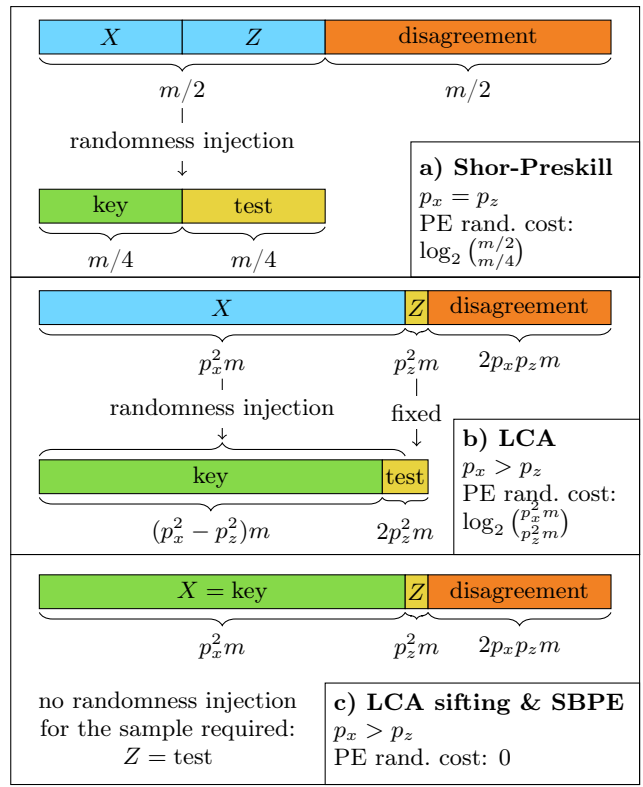
\begin{figure}[htb]
  \centering
  \begin{tikzpicture}[yscale=0.9,xscale=1.9]
    \definecolor{agreement}{HTML}{70D6FF}
    \definecolor{disagreement}{HTML}{F28123}
    \definecolor{test}{HTML}{E8D33F}
    \definecolor{key}{HTML}{86D949}


    \draw (-0.2,-3.3) rectangle (4.2,0.7);

    \draw[fill=agreement] (0,0) rectangle node {$X$} (1,0.5);
    \draw[fill=agreement] (1,0) rectangle node {$Z$} (2,0.5);
    \draw[decorate,decoration={brace,amplitude=5pt,mirror,raise=0.05cm}] 
      (0,0) -- (2,0) node[midway,yshift=-0.5cm] {$m$/2};

    \draw[fill=disagreement] (2,0) rectangle node {disagreement} (4,0.5);
    \draw[decorate,decoration={brace,amplitude=5pt,mirror,raise=0.05cm}] 
      (2,0) -- (4,0) node[midway,yshift=-0.5cm] {$m$/2};

    \draw[->] (1,-0.8) -- node[fill=white] {randomness injection} (1,-1.8);

    \draw[fill=key] (0,-2.5) rectangle node {key} (1,-2);
    \draw[decorate,decoration={brace,amplitude=5pt,mirror,raise=0.05cm}] 
      (0,-2.5) -- (1,-2.5) node[midway,yshift=-0.5cm] {$m$/4};

    \draw[fill=test] (1,-2.5) rectangle node {test} (2,-2);
    \draw[decorate,decoration={brace,amplitude=5pt,mirror,raise=0.05cm}] 
      (1,-2.5) -- (2,-2.5) node[midway,yshift=-0.5cm] {$m$/4};
      
    \draw (2.6,-3.3) rectangle (4.2,-1.3);
    \node[anchor=north west] at (2.6,-1.4) 
      {\parbox{4cm}{\raggedright \textbf{a) Shor-Preskill} \\ $p_x = p_z$ \\
	PE rand.\ cost: \\ $\log_2 {m/2 \choose m/4}$}};


    \draw (-0.2,-7.4) rectangle (4.2,-3.3);

    \draw[fill=agreement] (0,-4) rectangle node {$X$} (2.56,-3.5);
    \draw[decorate,decoration={brace,amplitude=5pt,mirror,raise=0.05cm}] 
      (0,-4) -- (2.56,-4) node[midway,yshift=-0.5cm] {$p_x^2m$};

    \draw[fill=test] (2.56,-4) rectangle node {$Z$} (2.72,-3.5);
    \draw[decorate,decoration={brace,mirror,raise=0.05cm}] 
      (2.56,-4) -- (2.72,-4) node[midway,yshift=-0.5cm] {$p_z^2m$};

    \draw[fill=disagreement] (2.72,-4) rectangle node {disagreement} (4,-3.5);
    \draw[decorate,decoration={brace,amplitude=5pt,mirror,raise=0.05cm}] 
      (2.72,-4) -- (4,-4) node[midway,yshift=-0.5cm] {$2p_xp_zm$};

    \draw[->] (1.28,-4.8) -- node[fill=white] {randomness injection} (1.28,-5.7);
    \draw[->] (2.64,-4.8) -- node[fill=white] {fixed} (2.64,-5.8);

    \draw[fill=key] (0,-6.5) rectangle node {key} (2.4,-6);
    \draw[decorate,decoration={brace,amplitude=5pt,mirror,raise=0.05cm}] 
      (0,-6.5) -- (2.4,-6.5) node[midway,yshift=-0.5cm] {$(p_x^2-p_z^2)m$};
    \draw[decorate,decoration={brace,amplitude=5pt,raise=0.05cm}] 
      (0,-6) -- (2.56,-6) node[midway,yshift=-0.5cm] {};

    \draw[fill=test] (2.4,-6.5) rectangle node {test} (2.72,-6);
    \draw[decorate,decoration={brace,mirror,raise=0.05cm}] 
      (2.4,-6.5) -- (2.72,-6.5) node[midway,yshift=-0.5cm] {$2p_z^2m$};
      
    \draw (2.9,-7.4) rectangle (4.2,-5.4);
    \node[anchor=north west] at (2.9,-5.4) 
      {\parbox{4cm}{\raggedright \textbf{b) LCA} \\ $p_x > p_z$ \\
	PE rand.\ cost: \\ $\log_2 {p_x^2m \choose p_z^2m}$}};


    \draw (-0.2,-10.8) rectangle (4.2,-7.4);

    \draw[fill=key] (0,-8.2) rectangle node {$X=\text{key}$} (2.56,-7.7);
    \draw[decorate,decoration={brace,amplitude=5pt,mirror,raise=0.05cm}] 
      (0,-8.2) -- (2.56,-8.2) node[midway,yshift=-0.5cm] {$p_x^2m$};

    \draw[fill=test] (2.56,-8.2) rectangle node {$Z$} (2.72,-7.7);
    \draw[decorate,decoration={brace,mirror,raise=0.05cm}] 
      (2.56,-8.2) -- (2.72,-8.2) node[midway,yshift=-0.5cm] {$p_z^2m$};

    \draw[fill=disagreement] (2.72,-8.2) rectangle node {disagreement} (4,-7.7);
    \draw[decorate,decoration={brace,amplitude=5pt,mirror,raise=0.05cm}] 
      (2.72,-8.2) -- (4,-8.2) node[midway,yshift=-0.5cm] {$2p_xp_zm$};
      
    \draw (2,-10.8) rectangle (4.2,-9.3);
    \node[anchor=north west] at (2,-9.3) 
      {\parbox{4cm}{\raggedright \textbf{c) LCA sifting \& SBPE} \\ $p_x > p_z$ \\
	PE rand.\ cost: 0 }};
    \node[anchor=north] at (0.9,-9.3) {\parbox{3.7cm}{no randomness injection for
      the sample required: $Z=\text{test}$}};
  \end{tikzpicture}
  \caption{Comparison of the expected sifting efficiencies. \textbf{a)} In the
  protocol of Shor and Preskill \cite{SP00}, only about a quarter of the
  measurement results end up in the raw key. Moreover, a relatively large amount
  of randomness needs to be injected for the sample choice, which in turn
  increases the length of pre-shared secret key that Alice and Bob use for
  authenticated communication. \textbf{b)} The protocol by Lo, Chau and
  Ardehali \cite{lca05} allows for a bias, $p_x>p_z$. This way, the expected
  fraction of bits with basis disagreements shrinks from one half to $2p_xp_z$.
  The proportions drawn in this figure correspond to $p_x=0.8$. However, it
  still requires randomness injection for the choice of the sample. \textbf{c)}
  If, instead, LCA sifting and SBPE are used, as we suggest, then no randomness
  injection is required for the choice of the sample. Moreover, less bits are
  consumed for parameter estimation in the finite-key regime, resulting in a
  longer raw key.
    \label{fig:eff-bars}}
\end{figure}

This scheme trivially has no basis information leak. In addition, it trivially
samples uniformly, as the whole sample is chosen with fresh randomness that is
injected for that purpose. Thus, it is secure with respect to the concerns
raised in this article. However, it is unnecessarily inefficient: speaking in
expectation values, half of the bits are discarded because they were determined
in different bases, and another quarter of the bits is used for parameter
estimation, leaving only a quarter of the original bits for the raw key, see
\Cref{fig:eff-bars} a).

A similar protocol has recently been suggested by Tomamichel and Leverrier with
a complete proof of its security, modelling all its subroutines \cite{TL15}.
They also use symmetric basis choice probabilities $p_x = p_z$ and randomness
injection for the sample choice. However, they do not use half of the sifted
bits for parameter estimation but less. Their protocol also samples uniformly,
because additional randomness is injected for the choice of the sample.

To increase the efficiency, Lo, Chau and Ardehali (LCA) suggested to use
asymmetric basis choice probabilities and quota, i.e. $p_x > 0$ and $k \neq n$.
As shown in \Cref{fig:eff-bars} b), this decreases the number of expected
disagreements from a value of $m/2$ to a value of $2p_xp_zm$. This is great for
efficiency: for larger block lengths, relatively smaller samples are required to
gain the same confidence that Alice's and Bob's bits are correlated.\footnote{
  This can be seen from inequality \eqref{eq:ptail-bound}, for example.} In 
the limit where $m \to \infty$, the probability $p_x$ can be chosen to be
arbitrarily close to one, and the fraction of data lost due to basis
disagreements converges to zero. We call this protocol \emph{LCA sifting}. It
shares property (1) with iterative sifting.

As for the protocol of Shor-Preskill, Lo Chau and Ardehali did not consider
SBPE. Their parameter estimation also requires some randomness injection for the
choice of the sample: the $Z$-agreements form one half of the sample, and the
other half is chosen at random from the $X$-agreements. Then, not just one but
two error rates are determined, namely on the $X$-part and the $Z$-part of the
sample separately. Only if \emph{both} error rates are below a fixed error
tolerance, they continue the protocol using the rest as the raw key (for
details, see their article \cite{lca05}). The LCA protocol trivially has no
basis information leak. In addition, it turns out that it also samples
uniformly. This is in fact non-trivial, and to our knowledge, it was not proved
in the literature. We fill this gap: the uniform sampling property of the LCA
protocol turns out to be a corollary of \Cref{prop:uniform-sampling} below.
Thus, the LCA protocol could be used as a secure replacement for iterative
sifting.

\begin{table}
  \textbf{LCA Sifting}
  \smallskip
  \hrule
  \begin{description}
    \item[Protocol Parameters] $n,k,m \in \posint$ with $m \geq n+k \in \posint$
      and $p_x, p_z \in [0,1]$ with $p_x + p_z = 1$.
    \item[Output] For $l=n+k$, the outputs are:
      \begin{itemize}
	\item[Alice:] $l$-bit string $\lstring{s} \in \{0,1\}^l$ (measurement
	  outcomes, sifted) or $s=\perp$ (if the protocol aborts),
	\item[Bob:] $l$-bit string $\lstring{t} \in \{0,1\}^l$ (measurement
	  outcomes, sifted) or $t=\perp$ (if the protocol aborts), 
	\item[public:] $l$-bit string $\lstring{\basis} \in \{0,1\}^l$ with
	  $\sum_i \basis_i = k$ (basis choices, sifted), where 0 means $X$-basis
	  and 1 means $Z$-basis, or $\basis = \perp$ (if the protocol aborts).
      \end{itemize}
    \item[Number of rounds] Fixed number $m$ (protocol parameter)
  \end{description}
  \smallskip
  \hrule
  \smallskip
  \textbf{The protocol}
  \begin{description}
    \item[Loop phase] Steps \ref{step:fixed-prep} to \ref{step:fixed-meas} are
      repeated $m$ times (round index $r=1,\ldots,m$). Starting with round
      $r=1$, Alice and Bob do the following:
      \begin{enumerate}[label=Step \arabic*:, ref=\arabic*]
	\item \textnormal{(Preparation):} \label{step:fixed-prep} Alice prepares
	  a qubit pair in a maximally entangled state.
	\item \textnormal{(Channel use):} Alice uses the quantum channel to send
	  one share of the qubit pair to Bob. \label{step:fixed-channel}
	\item \textnormal{(Random bit generation):} Alice and Bob each
	  (independently) generate a random classical bit $a_r$ and $b_r$,
	  respectively, where $0$ is generated with probability $p_x$ and $1$ is
	  generated with probability $p_z$. \label{step:fixed-random-bit}
	\item \textnormal{(Measurement):} \label{step:fixed-meas} Alice measures
	  her share in the $X$-basis (if $a_r = 0$) or in the $Z$-basis (if $a_r
	  = 1$), and stores the outcome in a classical bit $y_r$. Likewise, Bob
	  measures his share in the $X$-basis (if $b_r = 0$) or in the
	  $Z'$-basis (if $b_r = 1$), and stores the outcome in a classical bit
	  $y'_r$.
      \end{enumerate}
    \item[Final phase] The following steps are performed in a single run:
    \begin{enumerate}[label=Step \arabic*':, ref=\arabic*']
      \setcounter{enumi}{4}
      \item \textnormal{(Quota Check):} \label{step:quota-check} Alice and Bob
	determine the sets
	\begin{align*} 
	  &u(m) = \{ r \in \upto{m} \mid a_r = b_r = 0 \} \,, \\ 
	  &v(m) = \{ r \in \upto{m} \mid a_r = b_r = 1 \} 
	\end{align*}
	They check whether the quota condition ($u(m) \geq n$ and
	$v(m) \geq k$) holds. If it holds, they proceed with Step
	\ref{step:fixed-discarding}. Otherwise, they abort.
    \end{enumerate}
    \begin{enumerate}[label=Step \arabic*:, ref=\arabic*]
      \setcounter{enumi}{5}
      \item \textnormal{(Random Discarding):} \label{step:fixed-discarding}
	Alice and Bob choose a subset $u \subseteq u(m)$ of size $k$
	at random, i.e. each subset of size $k$ is equally likely to be chosen.
	Analogously, they choose a subset $v \subseteq v(m)$ of size
	$k$ at random.  Then they discard the bits $a_r$, $b_r$, $y_r$ and
	$y'_r$ for which $r \notin u \cup v$.
      \item \textnormal{(Order-preserving relabeling):} Let $r_i$ be the $i$-th
      element of $u \cup v$. Then Alice determines $\lstring{s} \in
      \{0,1\}^l$, Bob determines $\lstring{t} \in \{0,1\}^l$ and together they
      determine $\lstring{\basis} \in \{0,1\}^l$, where for every $i \in
      \upto{l}$, \begin{align*}
	s_i = y_{r_i} \,, \quad t_i = y'_{r_i} \,, \quad \basis_i = a_{r_i}
	\left( = b_{r_i} \right) \,.
      \end{align*}
      \item \textnormal{(Output):} \label{step:fixed-output} Alice locally
	outputs $\lstring{s}$, Bob locally outputs $\lstring{t}$ and they
	publicly output $\lstring{\basis}$.
  \end{enumerate}
  \end{description}
  \caption{The Lo-Chau-Ardehali (LCA) sifting protocol. 
    \label{prot:fixed-round}}
\end{table}

On the one hand, we suggest using the sifting part of LCA protocol. To be clear
about the details of the sifting scheme, we have written it out in our notation
in \Cref{prot:fixed-round}. On the other hand, we find that the parameter
estimation part of the LCA protocol is unnecessarily complicated and
inefficient: it needs randomness injection for part of the sample choice, and it
requires the estimation of two instead of one error rate. What if, instead, LCA
sifting is followed by SBPE, i.e., only the error rate on the $Z$-agreements is
determined? The critical question is whether this would still lead to uniform
sampling. As the following propositin shows, this is indeed the case.

\begin{prop}
  \label{prop:uniform-sampling}
  The combination of LCA sifting (\Cref{prot:fixed-round}) and SBPE
  (\Cref{prot:paramest}) samples uniformly. In other words, the LCA sifting
  protocol satisfies
  \begin{align}
    P_\Basis(\basis) = P_\Basis(\basis') \quad \Forall \basis, \basis' \in \lk
    \,.
  \end{align}
\end{prop}

In constrast to protocols that use randomness injection for the sample choice,
the uniform sampling property is non-trivial to prove for LCA sifting with SBPE.
We prove \Cref{prop:uniform-sampling} in \Cref{app:fixed-round} (see the
corollary of \Cref{prop:pbasis}). This shows that the combination of LCA sifting
and SBPE is secure and can therefore be used to replace iterative
sifting.\footnote{
  This also establishes uniform sampling for the whole LCA protocol (with the
  parameter estimation protocol with randomness injection instead of SBPE). This
  is because the parameter estimation protocol of LCA can now be seen as a
  two-stage random sampling without replacement, where in both stages, the
  sampling probabilities are uniform. This leads to overall uniform sampling.
}
For protocols that use these subroutines, the abort probability $\pabort$ of the
sifting step is important because it affects the key rate of the QKD protocol.
We calculate $\pabort$ in \Cref{app:fixed-round} as well (\Cref{prop:pbasis}).

This is good news for efficiency, as no randomness injection is required for the
choice of the sample. Since this random sample choice would need to be
communicated between Alice and Bob in an authenticated way, this also uses up
less secret key from the initial key pool (see \cite{FMC10} for a discussion of
the key cost of classical postprocessing). One can see in \Cref{fig:eff-bars}
that in the finite-key regime, this also leads to a larger raw key. Together
with \Cref{prop:tail}, which we will discuss in \Cref{sec:formal-criteria}, this
also establishes security of the protocol in the finite-key regime. In contrast,
the original work of LCA \cite{lca05} only establishes asymptotic security.

\begin{sug}
  Use LCA sifting (\Cref{prot:fixed-round}) and SBPE (\Cref{prot:paramest}).
\end{sug}

Let us briefly remark about the efficiency LCA sifting in
comparison to that of iterative sifting. They differ in that LCA sifting has no
communication during the loop phase, see property (3) above. The question is
whether this necessarily means that the efficiency is strongly reduced in
comparison with iterative sifting.

We define the efficiency $\eta$ of a sifting protocol as 
\begin{align}
  \label{eq:efficiency}
  \eta = \frac{R}{M} \,,
\end{align}
where $R$ is the random variable of the number of rounds that are kept after
sifting and $M$ is the random variable of the total number of rounds performed
in the loop phase of the protocol. We explain this in more detail in
\Cref{app:efficiency}. A plot of the expected efficiency for iterative sifting
and for LCA sifting is shown in \Cref{fig:efficiency} for the special case of
symmetric probabilities $p_x=p_z$ and identical quota $n=k$ (this
  special case is computationally much easier to calculate; for other choices,
the computation becomes very hard). We find that iterative sifting is more
efficient, as expected, but the difference between the two efficiencies becomes
insignificant for practically relevant quota sizes $n$ and $k$. 

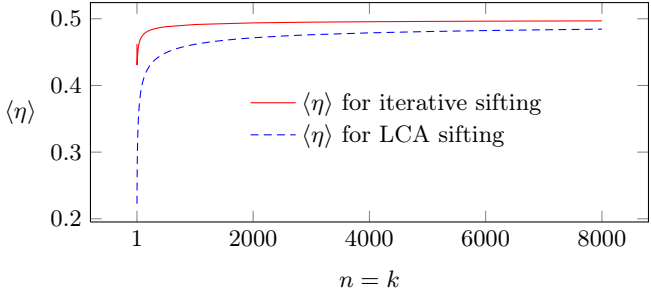
\begin{figure}[tbp]
  \centering
\begin{tikzpicture}
  \begin{axis}[
	  every axis legend/.append style={nodes={right}},
	  width=9cm,
	  height=4.5cm,
	  xtick={1,2000,4000,6000,8000},
	  xticklabels={1,2000,4000,6000,8000},
	  xlabel={$n=k$},
	  ylabel={$\langle\eta\rangle$},
	  ylabel style={rotate=-90, xshift=0.3cm},
	  legend entries={ {$\langle\eta\rangle$ for iterative sifting},
	  {$\langle\eta\rangle$ for LCA sifting}},
	    legend style={draw=none, at={(axis cs:4500, 0.35)}, anchor=center},
  ]
    \addplot[no markers, color=red] table[col sep=comma] {expEffI.csv};
    \addplot[no markers, color=blue, style=densely dashed] table[col sep=comma]
      {expEffF.csv};
  \end{axis}
\end{tikzpicture}
  \caption{Efficiency comparison of the two sifting protocols. The plots show
    lower bounds on the expected efficiencies for symmetric probabilities
    $p_x=p_z=1/2$ and for identical quotas $n=k$. The solid red curve shows a
    lower bound on the expected value of the efficiency for the iterative
    sifting protocol as a function of $n=k$. For the LCA sifting
    protocol, an optimization over the additional parameter $m$ has been made
    for each value of $n=k$.
    \label{fig:efficiency}}
\end{figure}

\section{Formal criteria for good sifting}
\label{sec:formal-criteria}

In \Cref{sec:problems}, we have seen that iterative sifting leads to problems.
In \Cref{sec:fix}, we showed that these problems can be avoided by using
LCA sifting (\Cref{prot:fixed-round}) and SBPE (\Cref{prot:paramest}). In this
section, we give a more complete answer to the question of how these
problems can be avoided by presenting two simple formal criteria that are
sufficient for a sifting protocol to lead to a correct parameter estimation.
More precisely, we describe two formal properties of the state produced by a
sifting protocol which guarantee that if the protocol is followed by SBPE
(\Cref{prot:paramest}), then Inequality \eqref{eq:ptail-bound} holds. As
indicated in the introduction, the two properties take the form of equalities,
see \Cref{eq:uniform-formal,eq:uncorr-formal}. We prove the sufficiency of these
two criteria by deriving \eqref{eq:ptail-bound} from them in \Cref{prop:tail}
below.

In order to state the two criteria and the random variable $\Lambdakey$ in
\eqref{eq:ptail-bound} formally, we need to define a certain kind of quantum
state $\rho_{A^lB^l\Basis^l}$ associated with a sifting protocol. To explain
what this state is, we explain what the state $\rho_{A^lB^l\Basis^l}$ is like
for LCA sifting. It is a state that is best described in a
variation of the protocol. Suppose that Alice and Bob run the protocol, but they
skip the measurement in every round.  Instead, they keep each qubit system in
their lab without modifying its state.  With current technology, this is
practically impossible, but since $\rho_{A^lB^l\Basis^l}$ is a purely
mathematical construct, we do not worry about the technical feasibility. Notice
that Alice and Bob still make basis choices, compare them and discard
rounds---they just do not actually perform the measurements. Let us compare the
output of this modified protocol with the output of the original protocol:
\renewcommand{\arraystretch}{1.3}
\begin{center}
  \begin{tabular}{rll} 
    & \quad original protocol & \quad modified protocol \\
    Alice: \quad & \quad $l$ bits $s=\lstring{s}$ & \quad $l$-qubit state
      $\rho_{A^l}$ \\
    Bob: \quad & \quad $l$ bits $t=\lstring{t}$ & \quad $l$-qubit state
      $\rho_{B^l}$ \\
    public: \quad & \quad $l$ bits $\basis=\lstring{\basis}$ & \quad $l$ bits
      $\basis=\lstring{\basis}$ 
  \end{tabular}
\end{center}
Hence, if we model the classical bit string $\basis$ as the state of a classical
register $\Basis^l$, we can say that the output of the modified protocol is a
quantum-quantum-classical (QQC) state $\rho_{A^lB^l\Basis^l}$. More generally,
the state $\rho_{A^lB^l\Basis^l}$ associated with a sifting protocol is its
output state in the case where all the measurements are skipped. 

This state still carries all the probabilistic information of the original
protocol. To see this, let $\bb{X} = \{\bb{X}_0, \bb{X}_1 \}$ and $\bb{Z} =
\{\bb{Z}_0, \bb{Z}_1\}$ be the POVMs describing Alice's $X$- and
$Z$-measurement, let $\bb{X}' = \{\bb{X}'_0, \bb{X}'_1 \}$ and $\bb{Z}' =
\{\bb{Z}'_0, \bb{Z}'_1\}$ be the POVMs describing Bob's $X$- and
$Z$-measurement, and let $\bb{M} = \{\bb{M}_0, \bb{M}_1\}$ be the projective
measurement on $\Basis$ with respect to which the state of the register $\Basis$
is diagonal. Define the operators
\begin{align}
  \begin{array}{llll}
    \bb{O}_0 = \bb{X}_0 \,, & \bb{O}_1 = \bb{X}_1 \,, & \bb{O}_2 = \bb{Z}_0 \,,
      & \bb{O}_3 = \bb{Z}_1 \,, \\
    \bb{O}'_0 = \bb{X}'_0 \,, & \bb{O}'_1 = \bb{X}'_1 \,, & \bb{O}'_2 =
      \bb{Z}'_0 \,, & \bb{O}'_3 = \bb{Z}'_1 \,.
  \end{array}
\end{align}
Then, the probability distribution over the output of the protocol is
\begin{align}
  P_{ST\Basis}(s,t,\basis) = \tr(\Pi(s,t,\basis) \rho_{(AB\Basis)^l}) \,,
\end{align}
where $\rho_{(AB\Basis)^l}$ is the same state as $\rho_{A^lB^l\Basis^l}$, but
with the registers reordered in the obvious way, and where
\begin{align}
  \Pi(s,t,\basis) = \bigotimes_{i=1}^l \left( \bb{O}_{2\basis_i+s_i} \otimes
  \bb{O}'_{2\basis_i+t_i} \otimes \bb{M}_{\basis_i} \right) \,.
\end{align}
With the state $\rho_{A^lB^l\Basis^l}$ associated with a sifting protocol at
hand, it is easy to define the random variable $\Lambdakey$ associated with the
protocol. The relevant probability space is the discrete probability space
$\left( \Omega_{ZZ'\Basis}, P_{ZZ'\Basis} \right)$, where $\Omega_{ZZ'\Basis}$
is the sample space
\begin{align}
  \label{eq:sample-space}
  \Omega_{ZZ'\Basis} &= \lcube\times\lcube\times\lk
\end{align}
and where $P_{ZZ'\Basis}$ is the probability mass function
\begin{align}
  \label{eq:pmf}
  \begin{array}{lccl}
    P_{ZZ'\Basis}: & \Omega_{ZZ'\Basis} & \rightarrow & [0,1] \\
    & (z,z',\basis) & \mapsto & \tr\Bigg( \left( \bigotimes_{i=1}^l \bb{Z}_{z_i}
    \right) \otimes \left( \bigotimes_{i=1}^l \bb{Z}'_{z'_i} \right) \\
    & & &\otimes
    \left( \bigotimes_{i=1}^l \bb{M}_{\basis_i} \right) \rho_{A^lB^l\Basis^l}
    \Bigg) \,.
  \end{array}
\end{align}
The probability mass function $P_{ZZ'\Basis}$ corresponds to a
\emph{Gedankenexperiment} in which Alice and Bob measure \emph{all} qubits in
the $Z$-basis.

Now we are able to formally say what the random variable $\Lambdakey$ of a
sifting protocol is. Let $\rho_{A^lB^l\Basis^l}$ be the state associated with
the sifting protocol, let $(\Omega_{ZZ'\Basis}, P_{ZZ'\Basis})$ be the
probability space as in \Cref{eq:sample-space,eq:pmf}. Then $\Lambdakey$ is the
random variable
\begin{align}
  \label{eq:lkey}
  \begin{array}{lccl}
    \Lambdakey: & \Omega_{ZZ'\Basis} & \rightarrow & [0,1] \\
    & (z,z',\basis) & \mapsto & {\displaystyle \frac{1}{n} \sum_{i=1}^N
    (1-\basis_i) (z \oplus z') } \,,
  \end{array}
\end{align}
which is the \emph{key bit error rate}. Analogously, we have the \emph{test bit
error rate}
\begin{align}
  \label{eq:ltest}
  \begin{array}{lccl}
    \Lambdatest: & \Omega_{ZZ'\Basis} & \rightarrow & [0,1] \\
    & (z,z',\basis) & \mapsto & {\displaystyle \frac{1}{k} \sum_{i=1}^l
    \basis_i (z \oplus z') } \,.
  \end{array}
\end{align}
This allows us to formally define the tail probability $\ptail$. We define it
via the same formula as in \eqref{eq:ptail-first}, which we repeat here for the
reader's convenience:
\begin{align}
  \label{eq:ptail}
  \ptail(\mu) = P[\Lambdakey \geq \Lambdatest + \mu \mid \Lambdatest \leq \qtol]
  \,. \tag{\ref{eq:ptail-first}}
\end{align}
The difference is that now, we have formally defined all the components of the
equality. The following proposition states the tail probability bound in a
formal way.

\begin{prop}[Tail probability estimate]
  \label{prop:tail}
  Let $\rho_{A^lB^l\Basis^l}$ be a density-operator of a system $A^l B^l
  \Basis^l$ where $A$ and $B$ are qubit systems and $\Basis$ is a classical
  system, let $\{\bb{Z}_0, \bb{Z}_1\}$ and $\{\bb{Z}'_0, \bb{Z}'_1\}$ be POVMs
  on the quantum systems $A$ and $B$, respectively, let $\{\bb{M}_0,\bb{M}_1\}$
  be the read-out measurement of the classical system $\Basis$, let
  $\Lambdakey$, $\Lambdatest$ be random variables on the discrete probability
  space $(\Omega_{ZZ'\Basis}, P_{ZZ'\Basis})$ as defined in
  \Cref{eq:sample-space,eq:pmf,eq:lkey,eq:ltest} and let $\ptail$ be as in
  \Cref{eq:ptail}. Let $\rho_{A^lB^l}$ and $\rho_{\Basis^l}$ denote the
  according reduced states of $\rho_{A^lB^l\Basis^l}$ and $P_{\Basis}$ denote
  the according marginal of $P_{ZZ'\Basis}$. If the two conditions
  \begin{align}
    &P_{\Basis}(\basis) = P_{\Basis}(\basis') \quad \Forall \basis, \basis' \in
      \lk \quad \text{and} \tag{\ref{eq:uniform-formal}} \\
    &\rho_{A^lB^l\Basis^l} = \rho_{A^lB^l} \otimes \rho_{\Basis^l}
      \tag{\ref{eq:uncorr-formal}}
  \end{align}
  hold, then
  \begin{align}
    \label{eq:ptail-bound2}
    \ptail(\mu) &\leq \frac{\exp\left( -2\frac{kn}{l}\frac{k}{k+1}\mu^2
  \right)}{\ppass} \,, \tag{\ref{eq:ptail-bound}}
  \end{align}
  where
  \begin{align}
    \ppass = P[\Lambdatest \leq \qtol] \,. \tag{\ref{eq:ppass}}
  \end{align}
\end{prop}

We prove \Cref{prop:tail} in \Cref{app:formal-criteria}. The formulation of
\Cref{prop:tail} allows us to see the formal requirements on a sifting protocol
to lead to a correct parameter estimation when followed by SBPE:
Condition \eqref{eq:uniform-formal} is exactly the statement that the sampling
probability does not depend on the sample, i.e. the protocol leads to uniform
sampling. There is one thing that we want to point out here: while it is
sufficient for the sampling probabilities to be the inverse of the number of
possible samples, i.e.
\begin{align}
  P_{\Basis}(\basis) = \frac{1}{\left\vert \lk \right\vert} = {l \choose k}^{-1}
    \quad \Forall \basis \in \lk \,, \end{align}
condition \eqref{eq:uniform-formal} is strictly weaker. In the case where there
is a non-zero probability that the protocol aborts during the sifting phase (as
it is the case for LCA sifting), the sampling probabilities do not
add up to 1 but rather to $1 - p_\text{abort}$, where $p_\text{abort}$ is the
probability that the protocol aborts during the sifting phase. 

Condition \eqref{eq:uncorr-formal} is the formal statement of what it means for
a protocol that the basis choice register is uncorrelated with Alice's and Bob's
qubits before measuring. \Cref{prop:tail} states that if these two conditions
are satisfied, then the correlation test of the SBPE protocol leads to
the right conclusion. Hence, these are the two conditions that a sifting
protocol needs to satisfy in order to be a good sifting protocol.

We point out that the digression to a classical probability space,
\Cref{eq:sample-space,eq:pmf,eq:lkey,eq:ltest,eq:ptail}, is a mere change of
notation. However, the fact that it is possible to express \Cref{prop:tail} in
terms of a classical probability space shows that this part of a QKD security
analysis is purely classical.

\section{Conclusion}
\label{sec:conclusion}

In recent years QKD has emerged as a commercial technology, with the prospect of
global QKD networks on the horizon \cite{SBP09}. All QKD implementations have
finite size, and yet only recently has finite-key analysis approached
mathematical rigor \cite{TLGR12, LCW14,Hayashi2011, Hayashi2014, CXC14,LPTRG13,
Furrer2011, Leverrier2014}. In this work, we showed that further
modifications of the protocols and/or their analysis are needed to make
finite-key analysis rigorous.

We pointed out that \textit{sifting}---a stage of QKD that is often overlooked
with respect to security analysis---is actually crucial for security. A
carelessly designed sifting subroutine can jeopardize the security of an
otherwise reliable protocol. We found that \emph{iterative sifting}, a sifting
protocol that has both been proposed theoretically
\cite{TLGR12,LPTRG13,CXC14,LCW14} and been implemented experimentally
\cite{BCLVV13}, violates two assumptions in the typical security analysis. We
showed how the violation of these assumptions can be exploited by an
eavesdropper, leading to intercept-resend attacks with unexpectedly low error
rates (see Fig.~\ref{fig:attack1}). 

We presented an alternative scheme, LCA sifting and SBPE, and proved
that it solves the two problems. We derived an expression for its abort
probability and therefore provided everything that is needed for
its future use as a subroutine. We argued that this scheme is more economical
and efficient than some other other previously proposed protocols, as it does not
require an additional random seed for the sample and at the same time allows for
asymmetric basis choice probabilities. As we explained, the latter allows for a
significantly higher sifting efficiency \cite{lca05}.

We gave the precise mathematical form of the two assumptions that are needed for
secure sifting in Eqs.~\eqref{eq:uniform-formal} and \eqref{eq:uncorr-formal}.
In doing so, we have provided a guide for the construction of future protocols:
when designing a sifting protocol, one just needs to check these two conditions
in order to make sure that the usual analysis of the parameter estimation based
on Inequality~\eqref{eq:ptail-bound} is correct and the protocol is secure. This
may require a mathematical model for the state $\rho_{A^lB^l\Basis^l}$ or for
the probabilities of the output strings $\lstring{\basis}$, $\lstring{s}$ and
$\lstring{t}$ generated by the sifting protocol. Such models are rarely provided
in the literature. In the case of iterative sifting, the absence of such a model
to check the desired properties has led to a wrong security analysis.

This points to a deeper problem in QKD security analysis: there is often a gap
between the physical protocols that are written down as instructions for Alice
and Bob and the mathematics of the security proof. This is not a purely pedantic
issue, but rather a very practical one which can be exploited by eavesdroppers.
In the future, we advocate that each step in the physical QKD protocol be
explicitly mathematically modeled. In particular, we emphasize that sifting
protocols must be proved to (rather than assumed to) satisfy the desired
assumptions of the analysis. We believe our work will ultimately inspire more
complete security proofs of finite-size QKD.

\section*{Acknowledgments}

We would like to thank Marco Tomamichel, David Elkouss and J{\k{e}}drzej
Kaniewski for insightful discussions. CP and SW are funded by Singapore's MOE
Tier 3A Grant MOE2012-T3-1-009, and STW, Netherlands. PJC and NL are supported
by Industry Canada, Sandia National Laboratories, NSERC Discovery Grant, and
Ontario Research Fund (ORF). 

\bibliographystyle{arxiv}
\bibliography{ArxivPaper}

\onecolumngrid
\pagebreak

\section*{Appendix}

\subsection*{Conventions}

We make some notational conventions for the appendix (in addition to the ones we
made in \Cref{sec:notation}). For the iterative sifting protocol as in
\Cref{prot:iterative}, we denote by $N_x$ the random variable of the number of
$X$-agreements, and analogously, $N_z$ and $N_d$ are the random variables of the
number of $Z$-agreements and disagreements in the loop phase, respectively. We
write events as logical statements of the random variables, e.g. $\Basis=110
\land N_x \geq 2$ is the event in which the protocol runs with more than two
$X$-agreements and produces the output $\basis=110$, and its probability is
given by $P[\Basis=110 \land N_x \geq 2]$. In cases where all involved random
variables have fixed values, we occasionally write expressions in terms of
probability mass functions instead of in terms of probability weights of events
(as we have done it in the main article), e.g. we write
\begin{align}
  P_{\Basis N_xN_zN_d}(\basis,n_x,n_z,n_d) := P[\Basis =\basis, N_x=n_x,
  N_z=n_z, N_d=n_d] \,. 
\end{align}
In cases with inequalities, it is however shorter to use the event notation,
e.g.
\begin{align}
  P[A_1 \neq B_1] = P_{A_1B_1}(0,1) + P_{A_1B_1}(1,0) \,.
\end{align}
We will use whatever notation we find more appropriate in each case.

\appendix

\section{Sampling probability calculation for iterative sifting}
\label{app:sample-prob}

In this appendix, we prove \Cref{prop:sample-prob}, i.e. we calculate the
sampling probabilities $P_\Basis(\basis)$ for iterative sifting with $n=1$ and
$k=2$ and find $P_\Basis (110) = g_z^2$ and $P_\Basis(101) = P_\Basis(011) =
(1-g_z^2)/2$, where $g_z = p_z^2 / (p_z^2 + p_x^2)$.

\begin{proof}[Proof of \Cref{prop:sample-prob}]
  We first write out the sequence of equalities that lead to the claim. We
  explain each equality below. The sequence of equalities looks as follows:
  \begin{align}
    P_\Basis (110) 
    &= \sum_{n_x=1}^\infty \sum_{n_z=2}^\infty \sum_{n_d=0}^\infty
      P_{\Basis N_xN_zN_d}(110,n_x,n_z,n_d) \label{eq:marginal} \\
    &= \sum_{n_z=2}^\infty \sum_{n_d=0}^\infty P_{\Basis
      N_xN_zN_d}(110,1,n_z,n_d) \label{eq:N_x=1} \\
    &= \sum_{n_z=2}^\infty \sum_{n_d=0}^\infty p_x^2 (p_z^2)^k (2p_xp_z)^d
      {n_z+n_d \choose n_d} \label{eq:insert-probs} \\
    &= g_z^2 \,, \quad \text{where } g_z = \frac{p_z^2}{p_z^2 + p_x^2} \,.
      \label{eq:evaluated}
  \end{align}
  Equation \eqref{eq:marginal} is just stating that $P_\Basis $ is the marginal
  of $P_{\Basis N_xN_zN_d}$. The ranges of the sums can be explained as follows.
  The iterative sifting protocol always runs until there have been at least $n$
  $x$-agreements and at least $k$ $z$-agreements. Therefore,
  \begin{align}
    P_{\Basis N_xN_zN_d}(\theta,n_x,n_z,n_d) = 0 \quad \text{if } n_x < n
    \text{ or } n_z < k \,.
  \end{align}
  In our case, $n=1$ and $k=2$, hence the limits of the sums.
  
  Equation \eqref{eq:N_x=1} follows from
  \begin{align}
    \label{eq:vanish}
    P_{\Basis N_xN_zN_d}(110,n_x,n_z,n_d) = 0 \quad \text{for } n_x \geq 2 \,.
  \end{align}
  One can see \eqref{eq:vanish} as follows: if $N_x \geq 2$, then necessarily
  $N_z=2$, because $N_x>n \land N_z>k$ is impossible in iterative sifting (the
  loop phase of the protocol is terminated as soon as both quota are met).  This
  means that during the random discarding, no $Z$-agreement gets discarded.
  Moreover, if $N_x \geq 2$, then the last round of the loop phase must be a
  $Z$-agreement. Since this $Z$-agreement is not discarded, we have that $\Basis
  $ must necessarily end in a $1$ if $N_x \geq 2$, so $\Basis =110$ is
  impossible in that case. 
  
  To see why Equation~\eqref{eq:insert-probs} holds, note that the event
  \begin{align}
    \Basis =110 \land N_x=1 \land N_z=n_z \land N_d=n_d
  \end{align}
  consists of all runs of the protocol in which one $X$-agreement, $n_z$
  $Z$-agreements and $n_d$ disagreements occurred, and where the $X$-agreement
  was the last round of the loop phase. This is because in every such run, one
  necessarily ends up with $\Basis =110$, and if $\Basis =110$, then the last
  round of the loop phase must be an $X$-agreement. There are ${n_z+n_d \choose
  n_d}$ such runs, and each of them has the probability $p_x^2 (p_z^2)^{n_z}
  (2p_xp_z)^{n_d}$, and therefore
  \begin{align}
    P_{\Basis N_xN_zN_d}(110,1,n_z,n_d) = {n_z+n_d \choose n_d}
    p_x^2 (p_z^2)^{n_z} (2p_xp_z)^{n_d} \,.
  \end{align}
  This explains Equation~\eqref{eq:insert-probs}. Finally, equation
  \eqref{eq:evaluated} is just an evaluation of the expression in the line
  above. This shows $P_\Basis (110) = g_z^2$. 
  
  It remains to be shown that $P_\Basis(101) = P_\Basis(011) = (1-g_z^2)/2$. In
  analogy to the above, it holds that
  \begin{align}
    P_\Basis (101) 
    &= \sum_{n_x=1}^\infty \sum_{n_z=2}^\infty \sum_{n_d=0}^\infty
      P_{\Basis N_xN_zN_d}(101,n_x,n_z,n_d) \label{eq:marginal2} \\
    &= \sum_{n_x=2}^\infty \sum_{n_d=0}^\infty P_{\Basis
      N_xN_zN_d}(101,n_x,2,n_d) \label{eq:fix-x} \,.
  \end{align}
  Equation \eqref{eq:marginal2} is, in analogy to Equation \eqref{eq:marginal},
  stating that $P_\Basis $ is the marginal of $P_{\Basis N_xN_zN_d}$, and the
  same argumentation for the limits of the sums applies. Equation
  \eqref{eq:fix-x} is explained by a similar reasoning as for
  Equation~\eqref{eq:N_x=1}: it follows from
  \begin{align}
    \label{eq:vanish2}
    P_{\Basis N_xN_zN_d}(101,n_x,n_z,n_d) = 0 \quad \text{for } n_z \geq 3 \,.
  \end{align}
  For Equation~\eqref{eq:vanish2}, note that if $N_z\geq3$, then $N_x=1$
  because $N_x>n \land N_z>k$ is impossible in iterative sifting. Thus, no
  $x$-agreement gets discarded. Moreover, if $N_z\geq3$, then the last round of
  the loop phase must be an $x$-agreement. Since this $x$-agreement is not
  discarded, $\Basis$ necessarily ends in a $0$ if $N_z\geq3$, so $\Basis=101$
  is impossible in this case. 

  Analogously, it holds that
  \begin{align}
    P_\Basis (011) 
    &= \sum_{n_x=1}^\infty \sum_{n_z=2}^\infty \sum_{n_d=0}^\infty
      P_{\Basis N_xN_zN_d}(011,n_x,n_z,n_d) \label{eq:marginal3} \\
    &= \sum_{n_x=2}^\infty \sum_{n_d=0}^\infty P_{\Basis
      N_xN_zN_d}(011,n_x,2,n_d) \label{eq:fix-x3} \,.
  \end{align}
  The next step is to realize that for every $n_x\geq2$ and for every $n_d \in
  \{0,1,2,\ldots\}$, it holds that
  \begin{align}
    P_{\Basis N_xN_zN_d}(101,n_x,2,n_d) = 
    P_{\Basis N_xN_zN_d}(011,n_x,2,n_d) \,. \label{eq:symm}
  \end{align}
  This is because the event
  \begin{align}
    (\Basis=101, N_x=n_x, N_z=2, N_d=n_d)
  \end{align}
  and the event
  \begin{align}
    (\Basis=011, N_x=n_x, N_z=2, N_d=n_d)
  \end{align}
  consist of equally many histories of the protocol, and each of these histories
  has the same probability. \Cref{eq:fix-x,eq:fix-x3,eq:symm} imply $P_\Basis
  (101) = P_\Basis (011)$. Since $P_\Basis (011) + P_\Basis (101) + P_\Basis
  (110) = 1$ and $P(110) = g_z^2$, it holds that $P_\Basis (011) = P_\Basis
  (101) = (1-g_z^2)/2$ as claimed.
\end{proof}

\newpage

\section{Error rate calculations for the attacks on iterative sifting}

\subsection{Attack that exploits non-uniform sampling}
\label{app:attack1}

Here, we calculate the expected error rate for the attack on iterative sifting
which exploits non-uniform sampling, as explained in
\Cref{sec:attack-on-non-uniform}. We first recall the relevant conventions that
we made in the main article. The iterative sifting protocol is described in
\Cref{prot:iterative}. Eve performs an intercept-resend attack during the loop
phase of the protocol. In the first round, she attacks in the $X$-basis, and in
all the other rounds of the loop phase, she attacks in the $Z$-basis. We defined
the error rate in Equation \eqref{eq:error-rate} in the main article, namely
\begin{align}
  E = \frac{1}{l} \sum_{i=1}^l S_i \oplus T_i \,.
\end{align}
Moreover, recall that we assume that the $X$- and $Z$-basis is the
same for Alice, Bob and Eve, and that they are mutually unbiased. This way, if
Alice and Bob measure in the same basis, but Eve measures in the other basis,
then Eve introduces an error probability of $1/2$ on this qubit. 

The calculation of $\langle E \rangle$ for this attack goes as follows. We first
make a split: 
\begin{align}
  \expect{E} 
  &= \sum_{\basis} P[\Basis =\basis]\expect{E|\Basis =\basis} \\
  &= \underbrace{P[\Basis =01]\expect{E|\Basis =01}}_{\Delta_x} + 
  \underbrace{P[\Basis =10]\expect{E|\Basis =10}}_{\Delta_z} \,.
  \label{eq:expect-e}
\end{align}
We have that
\begin{align}
  \Delta_x 
  &= \sum_{n_x=1}^\infty \bigg( P[\Basis =01 \land N_x=n_x \land A_1=B_1=0]
    \expect{E|\Basis =01 \land N_x=n_x \land A_1=B_1=0} \nonumber \\ 
  &\hspace{1.5cm} + P[\Basis =01 \land N_x=n_x \land A_1 \neq B_1]
    \expect{E|\Basis =01 \land N_x=n_x \land A_1 \neq B_1}
    \label{eq:expand-delta_x} \\
  &\hspace{1.5cm} + \underbrace{P[\Basis =01 \land N_x=n_x \land A_1=B_1=1]}_{0}
    \expect{E|\Basis =01 \land N_x=n_x \land A_1=B_1=1} \bigg) \nonumber
\end{align}
The third summand on the right hand side of Equation~\eqref{eq:expand-delta_x}
vanishes because $\Basis=01$ is impossible if Alice and Bob have a $z$-agreement
in the first round of the loop phase. The event
\begin{align}
  \Basis =01 \land N_x=n_x \land A_1=B_1=0 \label{eq:x-event}
\end{align}
consists of all histories of the protocol in which Alice and Bob have an
$x$-agreement in the first round and $n_x$ $x$-agreements in total. Infinitely
many such histories are possible because an arbitrary number of disagreements is
possible. We express the probability of the event~\eqref{eq:x-event} as the
marginal of the probability of the event
\begin{align}
  \Basis =01 \land N_x=n_x \land A_1=B_1=0 \land N_d=n_d \,. \label{eq:xd-event}
\end{align}
The event~\eqref{eq:xd-event} consists of ${n_x+n_d+1 \choose n_d}$ histories of
the protocol, and each history has the probability $(p_x^2)^{n_x} p_z^2
(2p_xp_z)^{n_d}$. Therefore,
\begin{align}
  P[\Basis =01 \land N_x=n_x \land A_1=B_1=0] 
  &= \sum_{n_d=0}^\infty P[\Basis =01 \land N_x=n_x \land A_1=B_1=0 \land
    N_d=n_d] \\
  &= \sum_{n_d=0}^\infty (p_x^2)^{n_x} p_z^2 (2p_xp_z)^{n_d} {n_x+n_d-1
  \choose n_d} \label{eq:first-prob}
\end{align}
Moreover, we have that
\begin{align}
  \expect{E|\Basis =01 \land N_x=n_x \land A_1=B_1=0} = \frac{1}{4} \left(
  1-\frac{1}{n_x} \right) \,. \label{eq:errrate-x}
\end{align}
The validity of \eqref{eq:errrate-x} can be seen as follows. On the second bit
of $S$ and $T$, there is no error because it comes from a round in which all
parties have measured in the $Z$-basis. Hence, the left had side of
\eqref{eq:errrate-x} is the probability of getting an error on the first bit of
$S$ and $T$, divided by the total number of bits, 2. Hence, we need to determine
the error probability of the first bit. If $N_x =1$, then the first bit comes
from the first round of the loop phase, in which Alice, Bob and Eve have
measured in the $X$-basis and hence, there is no error. However, for $N_x=n_x$,
the first bit of $S$ and $T$ is chosen at random from one of the $n_x$
$x$-agreements. In only one of these $n_x$ rounds, Eve has measured in the
$X$-basis, and in $n_x-1$ rounds, she measured in the $Z$-basis. Hence, the
probability that Eve measured in the wrong basis on the first bit of $S$ and $T$
is $(n_x-1)/n_x$, and therefore the error probability of the first bit is $1/2
\cdot (n_x-1)/n_x$. Thus,
\begin{align}
  \expect{E|\Basis =01 \land N_x=n_x \land A_1=B_1=0} &= \frac{1}{2} \cdot
  \frac{1}{2} \left( \frac{n_x}{n_x-1} \right) \\
  &= \frac{1}{4} \left( 1-\frac{1}{n_x} \right) \,.
\end{align}
Similarly, we get
\begin{align}
  P[\Basis =01 \land N_x=n_x \land A_1 \neq B_1] = \sum_{n_d=0}^\infty
    (p_x^2)^{n_x} p_z^2 (2p_xp_z)^{n_d} {n_x+n_d-1 \choose n_x}
  \label{eq:second-prob}
\end{align}
and
\begin{align}
  \expect{E|\Basis =01 \land N_x=n_x \land A_1 \neq B_1} = \frac{1}{4} \,.
  \label{eq:errrate-z}
\end{align}
Taking \Cref{eq:first-prob,eq:errrate-x,eq:second-prob,eq:errrate-z} together,
we get that
\begin{align}
  \Delta_x &= \frac{1}{4} \sum_{n_x=1}^\infty \sum_{n_d=0}^\infty (p_x^2)^{n_x}
    p_z^2 (2p_xp_z)^{n_d} \left( {n_x+n_d-1 \choose n_d}
    \left(1-\frac{1}{n_x}\right) + {n_x+n_d-1 \choose n_x} \right) \,.
    \label{eq:trianglex}
\end{align}
In a similar way, we get
\begin{align}
  \Delta_z = \frac{1}{4} \sum_{n_z=1}^\infty \sum_{n_d=0}^\infty p_x^2
  (p_z^2)^{n_z} (2p_xp_z)^{n_d} \left( {n_z+n_d-1 \choose n_d} + {n_z+n_d-1
  \choose n_d} \left( 1+\frac{1}{n_x} \right) \right) \,.
  \label{eq:trianglez}
\end{align}
\Cref{eq:expect-e,eq:trianglex,eq:trianglez} taken together result in
\begin{align}
  \expect{E} 
  &= \sum_{n_d=0}^\infty (2p_xp_z)^{n_d} \Bigg( \sum_{n_x=1}^\infty
    (p_x^2)^{n_x} p_x^2 \left( {n_x+n_d-1 \choose n_d} \left( 1-\frac{1}{n_x}
    \right) + {n_x+n_d-1 \choose n_x} \right) \nonumber \\
  &\hspace{1.3cm} + \sum_{n_z=1}^\infty p_x^2 (p_z^2)^{n_z} \left( {n_z+n_d-1
    \choose n_z} + {n_z+n_d-1 \choose n_d} \left( 1+\frac{1}{n_z} \right)
    \right) \Bigg) \,.
    \label{eq:attack1-e}
\end{align}
\Cref{fig:attack1} in the main article shows a plot of $\expect{E}$ as in
\eqref{eq:attack1-e} as a function of $p_x$. As one can see, $\expect{E}$
achieves a minimum of $\expect{E} \approx 22.8\%$ for $p_x \approx 0.73$.

\subsection{Attack that exploits basis-information leak}
\label{app:attack2}

Now we calculate the expected error rate of iterative sifting for the attack
which exploits basis-information leak as described in
\Cref{sec:attack-on-basis}. As before, let $\expect{E}$ be the expected value of
the error rate as defined in \Cref{eq:error-rate}. Again, we assume that the
$X$- and $Z$-basis are the same for Alice, Bob and Eve and that they are
mutually unbiased. Recall the strategy of Eve's intercept-resend attack: Before
the first round of the loop phase, Eve flips a fair coin. Let $F$ be the random
variable of the coin flip outcome and let $0$ and $1$ be its possible values. If
$F=0$, then in the first round, Eve attacks in the $X$ basis, and if $F=1$, she
attacks in the $Z$-basis.  In the subsequent rounds, she keeps attacking in that
basis until Alice and Bob first reached a basis agreement. If it is an
$X$-agreement (equivalent to $\Basis =01$), Eve attacks in the $Z$-basis in all
remaining rounds, and if it is a $Z$-agreement (equivalent to $\Basis =10$), she
attacks in the $X$-basis in all remaining rounds.

The calculation of $\expect{E}$ goes as follows: 
\begin{align}
  \expect{E} 
  &= P_F(0) \expect{E|F=0} + P_F(1) \expect{E|F=1} \label{eq:expect-e-split-up}
    \\ 
  &= \expect{E|F=0} \label{eq:symmetry} \\
  &= \underbrace{P_\Basis (01)}_{1/2} \expect{E|F=0 \land \Basis =01} +
    \underbrace{P_\Basis (10)}_{1/2} \underbrace{\expect{E|F=0 \land \Basis
    =10}}_{1/4} \,.  
    \label{eq:remains}
\end{align}
Equality \eqref{eq:expect-e-split-up} is just a decomposition of $\expect{E}$
into conditional expectations. Equality \eqref{eq:symmetry} follows from the
fact that the problem is symmetric under the exchange of $X$ and $Z$, i.e. under
the exchange of $0$ and $1$. The only quantity that is not trivial to calculate
in Equation \eqref{eq:remains} is the expected value of the error rate, given
that Eve first measures in $X$ and that the first basis agreement is an
$X$-agreement. It is calculated as follows:
\begin{align}
  \expect{E|F=0 \land \Basis =01} 
  &= \sum_{n_x=1}^\infty \expect{E|F=0 \land \Basis =01 \land N_x=n_x}
    \underbrace{P_{N_x|\Basis F}(n_x|01,0)}_{P_{N_x|\Basis }(n_x|01)} \\
  &= \sum_{n_x=1}^\infty \underbrace{\expect{E|F=0 \land \Basis =01 \land
    N_x=n_x}}_{\frac{n_x-1}{4n_x}}
    \underbrace{P_{N_x\Basis }(n_x,01)}_{\sum_{n_d=0}^\infty (p_x^2)^{n_x} p_z^2
    (2p_xp_z)^{n_d} {n_x+n_d \choose n_d}} \underbrace{\frac{1}{P_\Basis 
    (01)}}_{2} \\
  &= \sum_{n_x=1}^\infty \frac{n_x-1}{2n_x} \sum_{n_d=0}^\infty (p_x^2)^{n_x}
    p_z^2 (2p_xp_z)^{n_d} {n_x+n_d \choose n_d} \\
  &= \frac{1}{4}(1-\ln2) \,,
\end{align}
where $\ln$ denotes the logarithm to base $e$. Therefore,
\begin{align}
  \expect{E} 
  &= \frac{1}{2}\frac{1}{4}(1-\ln2) + \frac{1}{2}\frac{1}{4} \\
  &= \frac{2-\ln2}{8} \\
  &\approx 16.3\% \,.
\end{align}

\subsection{Attack that exploits both problems}
\label{app:attack3}

Here we present the error rate induced by the intercept-resend attack presented
in \Cref{sec:attack-on-both}, which exploits both non-uniform sampling and basis
information leak. Let us recall the attack strategy. In the first round of the
loop phase of the iterative sifting protocol, she attacks in the $X$-basis. She
keeps doing that in subsequent rounds until Alice and Bob announce a
basis-agreement. If they announce an $X$-agreement, Eve attacks in the $Z$-basis
in all the following rounds. Otherwise, she keeps attacking in the $X$-basis.

The calculation of the error rate is similar to the calculations done in
\Cref{app:attack1,app:attack2}. We only show the result here:

\begin{align}
  \label{eq:error-rate-both}
  \expect{E} = \sum_{n_z=1}^{\infty} \sum_{n_d=0}^{\infty}  p_x^{2} p_z^{2n_z}
  (2p_xp_z)^{n_d} \binom{n_z+n_d}{n_d} \frac{1}{4}+\sum_{n_x=1}^{\infty}
  \sum_{n_d=0}^{\infty}  p_x^{2n_x} p_z^{2} (2p_xp_z)^{n_d} \binom{n_x+n_d}{n_d}
  \frac{n_x-1}{4n_x} \,.
\end{align}
A plot of \eqref{eq:error-rate-both} is shown in \Cref{fig:attack1} as a
function of $p_x$. As one can see, the expected error rate has a minimum of
$\expect{E} \approx 15.8\%$ for $p_x \approx 0.57$. Hence, this combined attack
on both problems performs much better than the one on non-uniform sampling alone
(with a minimal expected error rate of $\approx 22.8\%$, see
\Cref{sec:attack-on-non-uniform}) and even better than the attack on the basis
information leak alone (with a minimal expected error rate of $\approx 16.3\%$,
see \Cref{sec:attack-on-basis}).

\newpage

\section{Sampling and abort probability calculation for LCA sifting}
\label{app:fixed-round}

In this appendix, we derive the general form of the probabillity distribution
$P_\Basis(\basis)$ for LCA sifting (\Cref{prot:fixed-round}) as a function of
the parameters $n$, $k$, $m$, $p_x$ and $p_z$. This achieves two goals: Firstly,
it turns out that the sampling probability $P_\Basis(\basis)$ is independent of
the sample $\basis \in \lk$, which shows that the protocol samples uniformly.
Secondly, we calculate the abort probability $\pabort = P_\Basis(\perp)$. This
abort probability influences the key rate of potential QKD protocols that use
this protocol as a subroutine, which makes $\pabort$ an important performance
parameter of the protocol.

We start by describing in \Cref{sec:on-prob-models} how we think that proofs of
sampling probabilities should be formalized and how the general strategy of our
proof looks like. In \Cref{app:prob-space-model,app:margin,app:derive-ptheta},
we show the proofs and finally derive $P_\Basis$.

\subsection{On probabilistic models of the protocol}
\label{sec:on-prob-models}

LCA sifting gives rise to a set $\Omega$ of histories of the
protocol. This set can be modelled as the set $\Omega =
\Omega_{ABYY'STUV\Basis}$ of all tuples
\begin{align}
  \label{eq:omega}
  \omega = (a,b,y,y',s,t,u,v,\basis) \,,
\end{align}
where each entry varies over all its possible values. There are finitely many
such histories, and each of them as a probability associated with it. This can
be expressed more formally in the language of discrete probability
theory\footnote{
  By \emph{discrete} probability theory, we mean probability theory with a
  discrete sample space $\Omega$, i.e. where $\Omega$ is finite or countably
  infinite.}
by saying that $\Omega$ forms the sample space of a discrete probability space
$(\Omega, P)$, on which a probability mass function $p$ is defined such that
$P(\omega)$ is the probability of a history $\omega$. Note that by choosing
$\Omega = \Omega_{AB\ldots\Basis}$, we also include impossible combinations of
$a$, $b$, \ldots, $\basis$. For example, a history $\omega$ as in
\eqref{eq:omega} with $u=v$ is not possible, because $u$ stands for the
$X$-agreements chosen for the raw key and $v$ stands for the $Z$-agreements
chosen for the sample, and the two cannot coincide. This is not a problem for
our model, because in this case, we simply have $P(\omega) = 0$.

In this probability theory language, the strings $a,b,\ldots,\basis$ are values
that random variables $A$, $B$, \ldots, $\Basis$ can take. Random variables are
maps from the sample space $\Omega$ to a set which is called the \emph{range} or
\emph{codomain} of the random variable. For example, the random variable $A$ is
a map
\begin{align}
  \begin{array}{cccc}
    A: & \Omega & \rightarrow & \range{A} \\
       & \omega & \mapsto     & A(\omega)
  \end{array}
\end{align}
where $\range{A}$ is the codomain of $A$. We denote the codomains of random
variables with calligraphic letters (except for the random variable $\Basis$,
whose codomain we denote by $\co{\Basis}$). According to the protocol, we have
\begin{align}
  \range{A} = \{0,1\}^m = \{ \mstring{a} \mid a_i \in \{0,1\} \Forall i \in
  \upto{m} \} \,.
\end{align}
In the case where we model
\begin{align}
  \Omega = \Omega_{ABYY'STUV\Basis} = \range{A} \times \range{B} \times \range{Y}
  \times \range{Y}' \times \range{S} \times \range{T} \times \range{U} \times
  \range{V} \times \co{\Basis} \,,
\end{align}
the random variables are simply the (set-theoretic) projections on the
respective components, e.g.
\begin{align}
  \begin{array}{cccc}
    A: & \Omega = \range{A} \times \range{B} \times \ldots \times \co{\Basis}
    & \rightarrow & \range{A} \,, \\
    & (a,b,\ldots,\basis) & \mapsto & a \,.
  \end{array}
\end{align}
Then, the probability $P_A(a)$ that $A = a$ is given by
\begin{align}
  \begin{array}{ccc>{\displaystyle}l}
    P_A: & \range{A} & \rightarrow & [0,1] \\
    & a & \mapsto & \phantom{=}\sum_{\omega \in A^{-1}(a)} P(\omega) \\
    & & & = \sum_{(b,y,\ldots,\basis)} P_{AB\ldots\Basis}(a,b,\ldots,\basis) \,,
  \end{array}
\end{align}
where we have written $P = P_{AB\ldots\Basis}$. This is because in the case
where $\Omega = \Omega_{AB\ldots\Basis}$, $P$ is simply the joint probability
distribution of the random variables $A$, $B$, \ldots, $\Basis$.

Setting $(\Omega, P) = (\Omega_{AB\ldots\Basis}, P_{AB\ldots\Basis})$ is
sufficient to describe the probabilities of the random variables $A$, $B$,
\ldots, $\Basis$ and functions thereof. For our purposes, however, this
description is overloaded. We do not need to incorporate all the random
variables $A$, $B$, \ldots, $\Basis$ in $\Omega$ and $P$. One reason is that
some of the random variables are completely determined by some of the other
random variables. For example, the string $s$ of Alice's sifted measurement
outcomes is completely determined by Alice's measurement outcomes $a$ and the
subsets $u$ and $v$. In the probability theory language, this is expressed as
the fact that the random variable $S$ is a function of the random variables $A$,
$U$ and $V$,
\begin{align}
  \label{eq:dep-1}
  S \equiv S(A,U,V) \,,
\end{align}
or more precisely,
\begin{align}
  \begin{array}{cccc}
    S: & \range{A} \times \range{U} \times \range{V} & \rightarrow & \range{S}
    \\
       & (a,u,v)                                     & \mapsto     & s(a,u,v)
  \end{array}
\end{align}
and its probability distribution is given by
\begin{align}
  P_S(s) &= \sum_{(a,u,v) \in S^{-1}(s)} P_{AUV}(a,u,v) \\
	 &= \sum_{\omega \in (S \circ A \times U \times V)^{-1}(s)} P(\omega)
	 \,.
\end{align}
There are more such dependencies in our list of random variables:
\begin{align}
  T      & \equiv T(B,U,V) \,, \label{eq:dep-2} \\
  \Basis & \equiv \Basis(U,V) \,. \label{eq:dep-3}
\end{align}
Hence, setting
\begin{align}
  \label{eq:indep-model}
  (\Omega,P) = (\Omega_{ABYY'UV},P_{ABYY'UV})
\end{align}
and using the dependencies \eqref{eq:dep-1}, \eqref{eq:dep-2} and
\eqref{eq:dep-3} leads to an equally powerful description, but with a smaller
probability space.

For our purposes, the space \eqref{eq:indep-model} is still overloaded. We are
only interested in the distribution $P_\Basis$ of $\Basis$. According to
\eqref{eq:dep-3}, the relevant probability space is $(\Omega_{UV}, P_{UV})$, and
$\Basis$ is a random variable
\begin{align}
  \begin{array}{cccc}
    \Basis: & \Omega_{UV} = \range{U} \times \range{V} & \rightarrow &
    \co{\Basis} \,, \\
            & (u,v) & \mapsto & \basis(u,v) \,.
  \end{array}
\end{align}
Then, $P_\Basis$ is given by
\begin{align}
  \begin{array}{ccc>{\displaystyle}c}
    P_\Basis: & \co{\Basis} & \rightarrow & [0,1] \\
    & \basis & \mapsto & \sum_{(u,v) \in \Basis^{-1}(\basis)}P_{UV}(u,v)
  \end{array}
\end{align}
It is difficult to write down the probability mass function $P_{UV}$ directly.
Instead, we will derive the probbility mass function $P_{ABUV}$ on the sample
space $\Omega_{ABUV}$, and arrive at the probability distribution $P_{UV}$ via
marginalization of $P_{ABUV}$:
\begin{align}
  P_{UV}(u,v) = \sum_{(a,b) \in \range{A} \times \range{B}} P_{ABUV}(a,b,u,v)
  \,.
\end{align}
Hence, the relevant probability space for our proof of uniform sampling of
LCA sifting is the probability space $(\Omega_{ABUV}, P_{ABUV})$.

\subsection{Formalization of $(\Omega_{ABUV}, P_{ABUV})$}
\label{app:prob-space-model}

According to what we said in the last subsection, the probability space that is
relevant for our proof of uniform sampling of LCA sifting is the
space $(\Omega_{ABUV},P_{ABUV})$, which describes the probabilities of the basis
choice strings $a$ and $b$ of Alice and Bob, as well as the choices $u$ and $v$
of the rounds that are used for the raw key and for parameter estimation,
respectively. We are going to formalize this space in this subsection.

We start by determining the sample space
\begin{align}
  \Omega_{ABUV} = \range{A} \times \range{B} \times \range{U} \times \range{V}
  \,.
\end{align}
In the loop phase of the protocol, Alice and Bob generate basis choice strings
$a = \mstring{a} \in \m$, $b = \mstring{b} \in \m$. This happens in every run,
no matter whether Alice and Bob abort the protocol in the final phase. Hence,
\begin{align}
  \label{eq:ab-range}
  \range{A} = \range{B} = \m \,.
\end{align}
In the final phase of the protocol, Alice and Bob do a quota check, in which
they determine the rounds in which both measured in the $X$-basis
($X$-agreement) the rounds in which both measured in the $Z$-basis
($Z$-agreements). In the case where they had less than $n$ $X$-agreements or
less than $k$ $Z$-agreements, they abort. In this case, Alice and Bob do not
choose subsets $u$ and $v$ of their $X$- and $Z$-agreements, respectively. We
model this by saying that in this case, $u = v = \perp$, where $\perp$ is just a
symbol indicating that Alice and Bob abort. In the case where the quota check of
the protocol is successful, Alice and Bob choose random subsets $u \subseteq
u(m)$ of size $n$ and $v \subseteq v(m)$ of size $k$. We represent these subets
by bit strings $u \in \mn$, $v \in \lk$, where
\begin{align}
  \mn = \left\{ \mstring{u} \in \m \ \middle\vert \ \sum_{i=1}^m u_i = n
    \right\} \,, \quad
  \mk = \left\{ \mstring{v} \in \m \ \middle\vert \ \sum_{i=1}^m v_i = k
    \right\} \,.  
\end{align}
They are to be interpreted as follows: For $u \in \mn$ and $i \in \upto{m}$,
$u_i = 1$ means that $i$ is contained in the subset $u \subseteq u(m)$, and
$u_i =0$ means that $i$ is not contained, and likewise for $v \in \mk$. The
requirement that the subsets $u$ and $v$ have size $n$ and $k$ translates into
the conditions that the string components sum up to $n$ and $k$, respectively.
Taking the two possibilities (the protocol aborts or the quota check is
successful) together, we have that
\begin{align}
  &\range{U} = \mn \cup \{ \perp \} \,, \label{eq:u-range} \\
  &\range{V} = \mk \cup \{ \perp \} \,. \label{eq:v-range} \,,
\end{align}
and hence
\begin{align}
  \label{eq:omegaabuv}
  \Omega_{ABUV} = \range{A} \times \range{B} \times \range{U} \times \range{V} =
  \m \times \m \times \Big(\mn \cup \{\perp\}\Big) \times \Big(\mk \cup
  \{\perp\}\Big) \,.
\end{align}
This is the sample space of the probability space $(\Omega_{ABUV}, P_{ABUV})$
that we are looking for. 

Next, we determine the probability mass function $P_{ABUV}$. We can write
\begin{align}
  P_{ABUV}(a,b,u,v) = P_{AB}(a,b) P_{UV|AB}(u,v|a,b) \label{eq:split-pabuv}
\end{align}
where $P_{UV|AB}(u,v|a,b)$ is the probability that $U=u$ and $V=v$, conditioned
on $A=a$ and $B=b$. The probability distribution $P_{AB}(a,b)$ is easily
determined. Each bit $a_i$, $b_i$, $i \in \upto{m}$ is generated independently
at random and takes the value $0$ with probability $p_x$ and the value $1$ with
probability $p_z$. Hence,
\begin{align}
  \Forall (a,b) \in \range{A} \times \range{B}: \quad P_{AB}(a,b) 
  &= \prod_{i=1}^m p_x^{1-a_i} p_z^{a_i} p_x^{1-b_i} p_z^{b_i} \\
  &= p_x^{m-\abs{a}} p_z^{\abs{a}} p_x^{m-\abs{b}} p_z^{\abs{b}} \\
  &= p_x^{2m-\abs{a}-\abs{b}} p_z^{\abs{a}+\abs{b}} \,, \label{eq:pab}
\end{align}
where for a string $a \in \m$, we write
\begin{align}
  \abs{a} := \sum_{i=1}^m a_i \,.
\end{align}
The conditional probability distribution $P_{UV|AB}$ is a bit more tricky to
write down. What is crucial for this conditional probability is whether the
strings $a$ and $b$ have at least $n$ $X$-agreements and at least
$Z$-agreements. We want to give this condition a formula as follows. Imagine
Alice and Bob want to count their $X$- and $Z$-agreements. To do so, they can
first determine the string $a \wedge b$, given by
\begin{align}
  a \wedge b := (a_ib_i)_{i=1}^m \,.
\end{align}
The $i$-th entry $a_ib_i$ of $a \wedge b$ is $1$ if the corresponding bits $a_i$
and $b_i$ are both $1$, i.e. if they had a $Z$-agreement, and $0$ otherwise.
Hence, to count their $Z$-agreements, they can sum up the components of $a
\wedge b$:
\begin{align}
  \text{number of $Z$-agreements } = \sum_{i=1}^m a_ib_i = \abs{a \wedge b} \,.
\end{align}
Therefore, the condition that Alice and Bob had at least $k$ $Z$-agreements can
be expressed as
\begin{align}
  \abs{a \wedge b} \geq k \,.
\end{align}
Likewise, the condition that they had at least $n$ $X$-agreements can be written
as
\begin{align}
  \abs{\no{a} \wedge \no{b}} \geq n \,,
\end{align}
where for a string $a \in \m$, we write
\begin{align}
  \no{a} = \mstring{1-a} \in \m \,.
\end{align}
Taken together, the quota check condition reads
\begin{align}
  \label{eq:quota-cond}
  \abs{\no{a} \wedge \no{b}} \geq n \quad \text{and} \quad \abs{a \wedge b} \geq
  k\,.
\end{align}
In the case where condition \eqref{eq:quota-cond} is not satisfied, Alice and
Bob abort, and therefore it must be that $(u,v) = (\perp,\perp)$. We can write
this as
\begin{align}
  \label{eq:first-range}
  \Forall (a,b) \in \m \times \m \text{ such that } (\abs{a \wedge b} < k \text{
  or } \abs{\no{a} \wedge \no{b}} < n): \quad 
  P_{UV|AB}(u,v|a,b) = \chi(u = v = \perp) \,,
\end{align}
where $\chi$ is the \emph{indicator} function, which evaluates to 1 if its
argument is true and which evaluates to 0 if its argument is false. 

For $(a,b) \in \m \times \m$ such that condition \eqref{eq:quota-cond} is
satisfied, the conditional probability $P_{UV|AB}$ is a little more difficult to
write down.  In that case, both $u = \perp$ and $v = \perp$ are impossible.
Moreover, only those $u \in \mn$ are possible which are subsets of Alice and
Bob's $X$-agreements, i.e. which satisfy
\begin{align}
  u_i = 1 \implies a_i = b_i = 0 \quad \Forall i \in \upto{m} \,.
\end{align}
Note that
\begin{align}
  \Forall (a,b,u) \in \m \times \m \times \mn: \quad
  (u_i = 1 \implies a_i = b_i = 0) \iff \abs{\no{a} \wedge \no{b} \wedge u} = n
  \,.
\end{align}
Hence, the condition that $u$ is a subset of the $X$-agreements simply reads
\begin{align}
  \abs{\no{a} \wedge \no{b} \wedge u} = n \,,
\end{align}
and likewise, the condition that $v$ is a subset of the $Z$-agreements reads
\begin{align}
  \abs{a \wedge b \wedge v} = k \,.
\end{align}
Hence, in the case where \eqref{eq:quota-cond} holds, only those $(u,v) \in \mn
\times \mk$ are possible for which
\begin{align}
  \label{eq:isconsistent}
  \abs{\no{a} \wedge \no{b} \wedge u} = n \quad \text{and} \quad \abs{a \wedge b
  \wedge v} = k \,.
\end{align}
We can combine the two conditions in a single formula:
\begin{align}
  \Forall (a,b,u,v) \in \m \times \m \times \mn \times \mk: \\ 
  (\abs{\no{a} \wedge \no{b} \wedge u} = n \text{ and } \abs{a \wedge b \wedge
  v} = k) \iff \abs{\no{a} \wedge \no{b} \wedge u} + \abs{a \wedge b \wedge v} =
  l \,,
\end{align}
where $l := n+k$. If this condition is satisfied, then the pair $u$ is a subset
of the $X$-agreements. Since the number of $X$-agreements is given by
$\abs{\no{a} \wedge \no{b}}$, we have that
\begin{align}
  \text{number of subsets of $X$-agreements of size $n$ } = {\abs{\no{a} \wedge
  \no{b}} \choose n} \,.
\end{align}
Since Alice and Bob are discarding surplus fully at random, each such subset is
equally likely, and thus, has a probability of $1 / {\abs{\no{a} \wedge \no{b}}
\choose n}$. Arguing similarly for $v$ and noting that the choices
of $u$ and $v$ are independent when the quota condition is passed leads to
\begin{align}
  \label{eq:second-range}
  \begin{array}{>{\displaystyle}l}
    \Forall (a,b) \in \m \times \m \text{ such that } \abs{a \wedge b} \geq k
    \text{ and } \abs{\no{a} \wedge \no{b}} \geq n: \\
    P_{UV|AB}(u,v|a,b) = \chi(u \neq \perp, v \neq \perp, \abs{\no{a} \wedge
    \no{b} \wedge u} + \abs{a \wedge b \wedge v} = l) {\abs{\no{a} \wedge
    \no{b}} \choose n}^{-1} {\abs{a \wedge b} \choose k}^{-1} \,.
  \end{array}
\end{align}
These two cases fully determine the conditional probability, i.e.
\eqref{eq:first-range} and \eqref{eq:second-range} determine $P_{UV|AB}$ for all
$(a,b) \in \m \times \m$, namely:
\begin{align}
  P_{UV|AB}(u,v|a,b) = 
  \begin{cases}
    \chi(u = v = \perp) & \text{if } \abs{a \wedge b} < k \text{ or }
      \abs{\no{a} \wedge \no{b}} < n \\
    \chi(u \neq \perp, v \neq \perp, \abs{\no{a} \wedge \no{b} \wedge u} +
      \abs{a \wedge b \wedge v} = l) {\abs{\no{a} \wedge \no{b}} \choose n}^{-1}
      {\abs{a \wedge b} \choose k}^{-1} & \text{if } \abs{a \wedge b} \geq k
      \text{ and } \abs{\no{a} \wedge \no{b}} \geq n
  \end{cases}
\end{align}
We can write this as
\begin{align}
  P_{UV|AB}(u,v|a,b) &= \chi(\abs{a \wedge b} < k \text{ or } \abs{\no{a} \wedge
    \no{b}} < n) \chi(u=v=\perp) \\
  &\phantom{=} + \chi(\abs{a \wedge b} \geq k \text{ and } \abs{\no{a} \wedge
    \no{b}} \geq n) \chi(u \neq \perp, v \neq \perp, \abs{\no{a} \wedge \no{b}
    \wedge u} + \abs{a \wedge b \wedge v} = l) {\abs{\no{a} \wedge \no{b}}
    \choose n}^{-1} {\abs{a \wedge b} \choose k}^{-1} \nonumber \\
  &= \chi(\abs{a \wedge b} < k \text{ or } \abs{\no{a} \wedge
    \no{b}} < n) \chi(u=v=\perp) \nonumber \\
  &\phantom{=} + \chi(u \neq \perp, v \neq \perp, \abs{\no{a} \wedge \no{b}
    \wedge u} + \abs{a \wedge b \wedge v} = l) {\abs{\no{a} \wedge \no{b}}
    \choose n}^{-1} {\abs{a \wedge b} \choose k}^{-1} \,, \label{eq:puvab}
\end{align}
where the last equality follows form
\begin{align}
  u \neq \perp, v \neq \perp, \abs{\no{a} \wedge \no{b} \wedge u} + \abs{a
  \wedge b \wedge v} = l \implies \abs{a \wedge b} \geq k \text{ and }
  \abs{\no{a} \wedge \no{b}} \geq n \,.
\end{align}
Taking \eqref{eq:split-pabuv}, \eqref{eq:pab} and \eqref{eq:puvab} together, we
get
\begin{align}
  P_{ABUV}(a,b,u,v) = p_x^{2m-\abs{a}-\abs{b}} p_z^{\abs{a}+\abs{b}} \Bigg(
  &\chi(\abs{a \wedge b} < k \text{ or } \abs{\no{a} \wedge \no{b}} < n)
    \chi(u=v=\perp) \nonumber \\
  &\phantom{=} + \chi(u \neq \perp, v \neq \perp, \abs{\no{a} \wedge \no{b}
    \wedge u} + \abs{a \wedge b \wedge v} = l) {\abs{\no{a} \wedge \no{b}}
    \choose n}^{-1} {\abs{a \wedge b} \choose k}^{-1} \Bigg) \,.
    \label{eq:pabuv}
\end{align}
This concludes our formalization of $(\Omega_{ABUV},P_{ABUV})$.

\begin{defn}
  We define the discrete probability space $(\Omega_{ABUV},P_{ABUV})$ by
  equations \eqref{eq:omegaabuv} and \eqref{eq:pabuv}.
\end{defn}

\subsection{Marginalization to $(\Omega_{UV}, P_{UV})$}
\label{app:margin}

\begin{defn}
  We define the probability space $(\Omega_{UV}, P_{UV})$ by
  \begin{align}
    \Omega_{UV} &:= \range{U} \times \range{V} = \Big(\mn \cup \{\perp\}\Big)
    \times \Big(\mk \cup \{\perp\}\Big) \\
    P_{UV}(u,v) &:= \sum_{a,b \in \range{A} \times \range{B}} P_{ABUV}(a,b,u,v)
    \,.
  \end{align}
\end{defn}

\begin{prop}
  \label{prop:margin}
  It holds that
  \begin{align}
    \label{eq:puv}
    &P_{UV}(u,v) = \chi(u=v=\perp) \left( \sum_{n_x=0}^{n-1}
      \sum_{n_z=0}^{m-n_x} + \sum_{n_x=n}^{m} \sum_{n_z=0}^{\min(m-n_x,k-1)}
      \right) {m \choose n_x} {m-n_x \choose n_z} 2^{m-n_x-n_z} p_x^{m+n_x-n_z}
      p_z^{m-n_x+n_z} \nonumber \\
    &\ + \chi(u \neq \perp, v \neq \perp, \abs{u \wedge v} = 0)
      \sum_{n_x=n}^{m-k} \sum_{n_z=k}^{m-n_x} {m-l \choose n_x-n} {m-k-n_x
      \choose n_z-k} 2^{m-n_x-n_z} p_x^{m+n_x-n_z} p_z^{m-n_x+n_z} {n_x \choose
      n}^{-1} {n_z \choose k}^{-1} \,.
  \end{align}
\end{prop}

\begin{proof}
  To show equation \eqref{eq:puv}, we need to show three things:
  \begin{align}
    &P_{UV}(\perp,\perp) = \left( \sum_{n_x=0}^{n-1} \sum_{n_z=0}^{m-n_x} +
    \sum_{n_x=n}^{m} \sum_{n_z=0}^{\min\{m-n_x,k-1\}} \right) {m \choose n_x}
    {m-n_x \choose n_z} 2^{m-n_x-n_z} p_x^{m+n_x-n_z} p_z^{m-n_x+n_z} \,,
    \tag{i} \\
    &\begin{array}{>{\displaystyle}l}
      \Forall (u,v) \in \mn \times \mk: \\
      P_{UV}(u,v) = \chi(\abs{u \wedge v} = 0) \sum_{n_x=n}^{m-k}
	\sum_{n_z=k}^{m-n_x} {m-l \choose n_x-n} {m-k-n_x \choose n_z-k}
	2^{m-n_x-n_z} p_x^{m+n_x-n_z} p_z^{m-n_z+n_z} {n_x \choose n} {n_z
	\choose k} \,, \tag{ii}
    \end{array}
    \\
    &\Forall (u,v) \in \Big( \{\perp\} \times \mk \Big) \cup \Big( \mn \times
    \{\perp\} \Big): P_{UV}(u,v) = 0 \,. \tag{iii}
  \end{align}
  We start with showing (i). We have that
  \begin{align}
    P_{UV}(\perp,\perp) 
    &= \sum_{(a,b) \in \range{A} \times \range{B}} P_{ABUV}(a,b,\perp,\perp) \\
    &= \sum_{(a,b) \in \range{A} \times \range{B}} p_x^{2m-\abs{a}-\abs{b}}
      p_z^{\abs{a}+\abs{b}} \chi(\abs{a \wedge b} < k \text{ or } \abs{\no{a}
      \wedge \no{b}} < n) \\
    &= \sum_{(a,b) \in \Gammabort} p_x^{2m-\abs{a}-\abs{b}}
      p_z^{\abs{a}+\abs{b}} \,,
  \end{align}
  where
  \begin{align}
    \Gammabort = \left\{ (a,b) \in \m \times \m \ \middle\vert \ \abs{a \wedge
    b} < k \text{ or } \abs{\no{a} \wedge \no{b}} < n \right\} \,.
  \end{align}
  We can partition $\Gammabort$ as follows:
  \begin{align}
    \Gammabort = \bigsqcup_{(n_x,n_z) \in \Iabort} \Gamma(n_x,n_z) \,,
  \end{align}
  where the ``square cup'' $\sqcup$ stands for disjoint union (the union of
  disjoint sets) and where
  \begin{align}
    \Iabort &= \{ (n_x,n_z) \in \{0,\ldots,m\} \times \{0,\ldots,m\} \mid
		  n_x+n_z\leq m, (n_x < n \text{ or } n_z < k) \} \,, \\
    \Gamma(n_x,n_z) &= \left\{ (a,b) \in \m \times \m \ \middle\vert \ \abs{a
    \wedge b} = n_x, \abs{\no{a} \wedge \no{b}} = n_z \right\} \,.
  \end{align}
  Hence,
  \begin{align}
    \label{eq:puv-comp}
    P_{UV}(\perp,\perp) = \sum_{(n_x,n_z) \in \Iabort} \sum_{(a,b) \in
    \Gamma(n_x,n_z)} p_x^{2m-\abs{a}-\abs{b}} p_z^{\abs{a}+\abs{b}}
  \end{align}
  The set $\Gamma(n_x,n_z)$ consists of all $(a,b) \in \m \times \m$ with
  exactly $n_x$ $X$-agreements and exactly $n_z$ $Z$-agreements. For these
  strings,
  \begin{align}
    \Forall (a,b) \in \Gamma(n_x,n_z): p_x^{2m-\abs{a}-\abs{b}}
    p_z^{\abs{a}+\abs{b}} 
    &= p_x^{2n_x}p_z^{2n_z}(p_xp_z)^{m-n_x-n_z} \\
    &= p_x^{m+n_x-n_z} p_z^{m-n_z+n_x} \,,
  \end{align}
  so equation \eqref{eq:puv-comp} simplifies to
  \begin{align}
    \label{eq:puv-simple}
    P_{UV}(\perp,\perp) = \sum_{(n_x,n_z)\in\Iabort} \abs{\Gamma(n_x,n_z)}
    p_x^{m+n_x-n_z} p_z^{m-n_z+n_x} 
  \end{align}
  The number $\abs{\Gamma(n_x,n_z)}$ of elements of $\Gamma(n_x,n_z)$ is given
  by
  \begin{align}
    \label{eq:gamma-count}
    \abs{\Gamma(n_x,n_z)} = {m \choose n_x} {m-n_x \choose n_z} 2^{m-n_x-n_z}
    \,.
  \end{align}
  This can be seen as follows: ${m \choose n_x}$ is the number of possible
  distributions of the $n_x$ $X$-agreements over the $m$ rounds, and ${m-n_x
  \choose n_x}$ is the number of possible distributions of the $n_z$
  $Z$-agreements over the remaining $m-n_x$ rounds. For the rounds where the
  strings have basis agreement, they are fully determined, but for $i$ in the
  remaining $m-n_x-n_z$ rounds, we can have that either $a_i=0$ and $b_i=1$ for
  a basis disagreement or $a_i=1$ and $b_i=0$. Thus, there are two possibilities
  for every disagreement, which explains the factor $2^{m-n_x-n_z}$.
  Combining equations \eqref{eq:puv-simple} and \eqref{eq:gamma-count} yields
  \begin{align}
    P_{UV}(\perp,\perp) 
    &= \sum_{(n_x,n_z)\in\Iabort} {m \choose n_x} {m-n_x \choose n_z}
    2^{m-n_x-n_z} p_x^{m+n_x-n_z} p_z^{m-n_z+n_x} \\
    &= \left( \sum_{n_x=0}^{n-1} \sum_{n_z=0}^{m-n_x} +
	\sum_{n_x=n}^{m} \sum_{n_z=0}^{\min(m-n_x,k-1)} \right) {m \choose n_x}
	{m-n_x \choose n_z} 2^{m-n_x-n_z} p_x^{m+n_x-n_z} p_z^{m-n_x+n_z} \,,
  \end{align}
  where the last equation follows from splitting up $\Iabort$ into the two
  respective sets. This shows (i).

  We proceed with showing (ii). We get from equation \eqref{eq:pabuv} that
  \begin{align}
    &\Forall (u,v) \in \mn \times \mk: \\
    P_{UV}(u,v) 
    &= \sum_{(a,b) \in \m \times \m} p_x^{2m-\abs{a}-\abs{b}}
      p_z^{\abs{a}+\abs{b}} \chi(u \neq \perp, v \neq \perp, \abs{\no{a} \wedge
      \no{b} \wedge u} + \abs{a \wedge b \wedge v} = l) {\abs{\no{a} \wedge
      \no{b}} \choose n}^{-1} {\abs{a \wedge b} \choose k}^{-1} \\
    &= \sum_{(a,b) \in \Phi(u,v)} p_x^{2m-\abs{a}-\abs{b}} p_z^{\abs{a}+\abs{b}}
      {\abs{\no{a} \wedge \no{b}} \choose n}^{-1} {\abs{a \wedge b} \choose
      k}^{-1} \,, 
  \end{align}
  where
  \begin{align}
    \Phi(u,v) = \{(a,b) \in \m \times \m \mid \abs{\no{a} \wedge \no{b} \wedge
    u} + \abs{a \wedge b \wedge v} = l \}
  \end{align}
  In analogy to the way we split up $\Gammabort$ above, we now split up
  $\Phi(u,v)$:
  \begin{align}
    \Phi(u,v) = \bigsqcup_{(n_x,n_z) \in \Ipass} \Phi(u,v,n_x,n_z) \,,
  \end{align}
  where
  \begin{align}
    \Ipass &= \{(n_x,n_z) \in \{0,\ldots,m\} \times \{0,\ldots,m\} \mid
      n_x+n_z \leq m, n_x \geq n, n_z \geq k \} \,, \\
    \Phi(u,v,n_x,n_z) &= \{(a,b) \in \m \times \m \mid \abs{\no{a} \wedge \no{b}
    \wedge u} + \abs{a \wedge b \wedge v} = l, \abs{\no{a} \wedge \no{b}} = n_x,
    \abs{a \wedge b} = n_z \} \,.
  \end{align}
  This gives us
  \begin{align}
    \label{eq:puv-phi}
    &\Forall (u,v) \in \mn \times \mk: P_{UV}(u,v) = \sum_{(n_x,n_z) \in \Ipass}
    \sum_{(a,b) \in \Phi(u,v,n_x,n_z)} p_x^{2m-\abs{a}-\abs{b}}
    p_z^{\abs{a}+\abs{b}} {\abs{\no{a} \wedge \no{b}} \choose n}^{-1} {\abs{a
    \wedge b} \choose k}^{-1}
  \end{align}
  Again, in analogy to our calculation of $P_{UV}(u,v)$, the sets
  $\Phi(u,v,n_x,n_z)$ are sets on which the summand in equation
  \eqref{eq:puv-phi} is constant. More precisely, for every $(a,b,u,v) \in
  \m \times \m \times \Ipass$, it holds that
  \begin{align}
    \Forall (a,b) \in \Phi(u,v,n_x,n_z): \quad p_x^{2m-\abs{a}-\abs{b}}
      p_z^{\abs{a}+\abs{b}} {\abs{\no{a} \wedge \no{b}} \choose n}^{-1} {\abs{a
      \wedge b} \choose k}^{-1} 
    &= p_x^{2n_x} p_z^{2n_z} (p_xp_z)^{m-n_x-n_z} {n_x \choose n} {n_z \choose
      k} \\
    &= p_x^{m+n_x-n_z} p_z^{m-n_z+n_z} {n_x \choose n} {n_z \choose k}
    \label{eq:constant-puv}
  \end{align}
  This leads us to determining the size of $\Phi(u,v,n_x,n_z)$. In words, this
  set contains all pairs $(a,b) \in \m \times \m$ with $n_x$ $X$-agreements and
  $n_z$ $Z$-agreements such that $n$ $X$-agreements are located where $u_i=1$
  and $k$ $Z$-agreements are located where $v_i=1$. The size of this set is
  \begin{align}
    \abs{\Phi(u,v,n_x,n_z)} = \chi(\abs{u \wedge v} = 0) {m-l \choose n_x-n}
    {m-k-n_x \choose n_z-k} 2^{m-n_x-n_z} \,. \label{eq:sizeofphi}
  \end{align}
  This can be seen as follows. If $\abs{u \wedge v} \neq 0$, there cannot be any
  $(a,b) \in \m \times \m$ such that $\abs{\no{a} \wedge \no{b} \wedge u} +
  \abs{a \wedge b \wedge v} = l$, and hence the set must be empty in that case.
  This explains the factor $\chi(\abs{u \wedge v} = 0)$. For those $(u,v) \in
  \mn \times \mk$ for which $\abs{u \wedge v} = 0$, the strings $(a,b) \in
  \Phi(u,v,n_x,n_z)$ are determined on $n+k=l$ positions by $u$ and $v$. On the
  remaining $m-l$ rounds are partitioned into $n_x-n$ rounds of $X$-agreements,
  $n_z-k$ $Z$-agreements and $m-n_x-n_z$ disagreements. There are ${m-l \choose
  n_x-n} {m-k-n_x \choose n_z-k}$ such partitions. Finally, on each position of
  the $m-n_x-n_z$ disagreements, we have the two possibilities $(a_i,b_i) =
  (0,1)$ and $(a_i,b_i) = (1,0)$, which contributes the factor $2^{m-n_x-n_z}$.
  Taking equations \eqref{eq:constant-puv} and \eqref{eq:sizeofphi} together, we
  get
  \begin{align}
    &\Forall (u,v) \in \mn \times \mk: \nonumber \\
    P_{UV}(u,v) 
    &= \sum_{(n_x,n_z) \in \Ipass} \chi(\abs{u \wedge v} = 0) {m-l \choose
      n_x-n} {m-k-n_x \choose n_z-k} 2^{m-n_x-n_z} p_x^{m+n_x-n_z}
      p_z^{m-n_z+n_z} {n_x \choose n} {n_z \choose k} \\
    &= \chi(\abs{u \wedge v} = 0) \sum_{n_x=n}^{m-k} \sum_{n_z=k}^{m-n_x} {m-l
      \choose n_x-n} {m-k-n_x \choose n_z-k} 2^{m-n_x-n_z} p_x^{m+n_x-n_z}
      p_z^{m-n_z+n_z} {n_x \choose n} {n_z \choose k} \,.
  \end{align}
  This shows (ii).
  
  The remaining case (iii) is easily shown. It follows directly from
  \eqref{eq:pabuv}, because
  \begin{align}
    &\Forall (u,v) \in \Big( \{\perp\} \times \mk \Big) \cup \Big( \mn \times
    \{\perp\} \Big):
    \chi(u=v=\perp) = \chi(u \neq \perp, v \neq \perp, \abs{\no{a} \wedge \no{b}
    \wedge u} + \abs{a \wedge b \wedge v} = l) = 0 \,.
  \end{align}
  This shows (iii) and therefore completes the proof.
\end{proof}

\subsection{Formalization of $\Basis$ and derivation of $P_\Basis$}
\label{app:derive-ptheta}

We have derived the probability space $(\Omega_{UV},P_{UV})$ as demanded in
\Cref{sec:on-prob-models}. Now we are left to define the random variable
\begin{align}
  \begin{array}{cccl}
    \Basis: & \Omega_{UV} & \rightarrow & \co{\Basis} \\
    & (u,v) & \mapsto & 
    \begin{cases}
      h(u,v) & \text{if } (u,v) \in \{(u,v) \in \mn \times \mk \mid \abs{u
      \wedge v} = 0 \} \,, \\
      \perp  & \text{otherwise} \,.
    \end{cases}
  \end{array}
\end{align}
and to derive an expression for
\begin{align}
  \label{eq:pbasis-general}
  \begin{array}{ccc>{\displaystyle}c}
    P_\Basis: & \co{\Basis} & \rightarrow & [0,1] \\
    & \basis & \mapsto & \sum_{(u,v) \in \Basis^{-1}(\basis)}P_{UV}(u,v) \,.
  \end{array}
\end{align}
The range $\co{\Basis}$ of $\Basis$ is given by
\begin{align}
  \co{\Basis} = \lk \cup \{\perp\} \,,
\end{align}
where an element of $\lk$ is a sifted basis choice string as in LCA sifting and
where we set $\theta = \perp$ in the case where Alice and Bob abort the
protocol.

To derive the random variable $\Basis$, assume that Alice and Bob arrived at
strings $(u,v) \in \range{U} \times \range{V}$. How do these two strings
determine the sifted basis choice string $\basis$? Let us first assume the case
where $(u,v) \in \mn \times \mk$ such that $\abs{u \wedge v} = 0$. The relevant
set of indices in this case is the set of round indices $r$ for which $u_r=1$ or
$v_r=1$:
\begin{align}
  \alpha(u,v) := \{ r \in \m \mid u_r = 1 \text{ or } v_r = 1 \} \,.
\end{align}
Note that $\abs{\alpha(u,v)} = n+k = l$. For $i \in \upto{l}$, we define
\begin{align}
  \label{eq:alpha}
  \alpha_i(u,v) := \text{the $i$-th element of } \alpha(u,v) \,.
\end{align}
With this notation at hand, we can determine $\basis$ from $u$ and $v$ as
follows: for $i \in \upto{l}$, we have that $\basis_i = 0$ if $u_{\alpha_i(u,v)}
= 1$ and $\basis_i = 1$ if $v_{\alpha_i(u,v)} = 1$. (Note that for $i \in
\upto{l}$, it always holds either $u_{\alpha_i(u,v)} = 1$ or $v_{\alpha_i(u,v)}
= 1$, but never both, so this is well-defined.) We can write this in terms of a
helper function $h$ as
\begin{align}
  \begin{array}{cccl}
    h: & \{(u,v) \in \mn \times \mk \mid \abs{u \wedge v} = 0 \} & \rightarrow &
    \lk \\
    & (u,v) & \mapsto & (h_i(u,v))_{i=1}^l \,,
  \end{array}
\end{align}
where
\begin{align}
  h_i(u,v) =
  \begin{cases}
    0 & \text{if } u_{\alpha_i(u, v)} = 1 \,, \\
    1 & \text{if } v_{\alpha_i(u, v)} = 1 \,.
  \end{cases}
\end{align}
This determines $\Basis$ for all $(u,v) \in \mn \times \mk$ such that $\abs{u
\wedge v} = 0$. However, since these are the only pairs $(u,v)$ for which a
sifted basis choice string $\basis \in \lk$ is generated, we just let $\Basis$
send all other pairs $(u,v)$ to $\perp$:
\begin{align}
  \label{eq:basis}
  \begin{array}{cccl}
    \Basis: & \range{U} \times \range{V} & \rightarrow & \co{\Basis} \\
    & (u,v) & \mapsto & 
    \begin{cases}
      h(u,v) & \text{if } (u,v) \in \{(u,v) \in \mn \times \mk \mid \abs{u
      \wedge v} = 0 \} \,, \\
      \perp  & \text{otherwise} \,.
    \end{cases}
  \end{array}
\end{align}
This way, pairs $(u,v)$ are mapped to $\perp$ which cannot occur in the protocol
(e.g. $(\perp,b)$ with $b \in \lk$). This is unproblematic, because for these
pairs, $P_{UV}(u,v)=0$, so according to equation \eqref{eq:pbasis-general}, they
do not contribute to $P_\Basis$.
\begin{defn}
  We define the sifted basis choice string random variable $\Basis$ on
  $\Omega_{UV}$ by equation \eqref{eq:basis}. Its associated probability mass
  function $P_\Basis$ is given by \eqref{eq:pbasis-general}.
\end{defn}

We are ready to state the result.

\begin{prop}
  \label{prop:pbasis}
  For LCA sifting (\Cref{prot:fixed-round}), we have
  that \begin{align}
    &\pabort = P_\Basis(\perp) = \left( \sum_{n_x=0}^{n-1}
      \sum_{n_z=0}^{m-n_x} + \sum_{n_x=n}^{m} \sum_{n_z=0}^{\min(m-n_x,k-1)}
      \right) {m \choose n_x} {m-n_x \choose n_z} 2^{m-n_x-n_z} p_x^{m+n_x-n_z}
      p_z^{m-n_x+n_z} \,, \label{eq:result-abort} \\
    &\Forall \basis \in \lk: P_\Basis(\basis) = {m \choose n+k}
    \sum_{n_x=n}^{m-k} \sum_{n_z=k}^{m-n_x} {m-n-k \choose n_x-n} {m-k-n_x
    \choose n_z-k} 2^{m-n_x-n_z} p_x^{m+n_x-n_z} p_z^{m-n_x+n_z} {n_x \choose
    n}^{-1} {n_z \choose k}^{-1} \,. \label{eq:result-sample}
  \end{align}
\end{prop}

Before we prove \Cref{prop:pbasis}, let us point out its importance.  Equation
\eqref{eq:result-abort} is the probability that the sifting protocol aborts
because Alice and Bob did not reach the quota on the $X$- and $Z$-agreements,
and is therefore a performance parameter of the protocol.  Equation
\eqref{eq:result-sample} is the sampling probability for each $\basis \in \lk$.
Since \eqref{eq:result-sample} is independent of $\basis \in \lk$, we get
uniform sampling as a corollary of \Cref{prop:pbasis}.

\begin{cor}
  \label{cor:uniform-sampling}
  The combination of LCA sifting (\Cref{prot:fixed-round}) and SBPE
  (\Cref{prot:paramest}) samples uniformly. In other words, the LCA sifting
  protocol satisfies
  \begin{align}
    P_\Basis(\basis) = P_\Basis(\basis') \quad \Forall \basis, \basis' \in \lk
    \,.
  \end{align}
\end{cor}

This proves \Cref{prop:uniform-sampling}. It leads us to proposing the protocol
as a secure alternative to the insecure iterative sifting protocol.

Now we proceed to the proof of \Cref{prop:pbasis}.

\begin{proof}[Proof of \Cref{prop:pbasis}]
  We first show equation \eqref{eq:result-abort}. By definition, it holds that
  \begin{align}
    P_\Basis(\perp) = \sum_{(u,v) \in \Basis^{-1}(\perp)} P_{UV}(u,v) \,,
  \end{align}
  where
  \begin{align}
    \Basis^{-1}(\perp) = \Big( \{\perp\} \times \mk \Big) \cup \Big( \mn \cup
    \{\perp\} \Big) \cup \{(\perp,\perp)\} \cup \{ (u,v) \in \mn \times \mk \mid
      \abs{u \wedge v} \neq 0 \}
  \end{align}
  We know from \Cref{prop:margin} that
  \begin{align}
    \Forall (u,v) \in \Big( \{\perp\} \times \mk \Big)
    \cup \Big( \mn \cup \{\perp\} \Big): \quad P_{UV}(u,v) = 0 \,.
  \end{align}
  Since
  \begin{align}
    \Forall (a,b,u,v) \in \m \times \m \times \mn \times \mk: \quad
    \abs{u \wedge v} \neq 0 \implies \abs{\no{a} \wedge \no{b} \wedge u} +
    \abs{a \wedge b \wedge v} \neq 0 \,,
  \end{align}
  we also have
  \begin{align}
    \Forall (u,v) \in \{(u',v') \in \mn \times \mk \mid \abs{u' \wedge v'} \neq
    0 \}: \quad P_{UV}(u,v) = 0 \,.
  \end{align}
  Thus,
  \begin{align}
    P_\Basis(\perp) &= P_{UV}(\perp,\perp) \\
    &= \left( \sum_{n_x=0}^{n-1} \sum_{n_z=0}^{m-n_x} + \sum_{n_x=n}^{m}
      \sum_{n_z=0}^{\min(m-n_x,k-1)} \right) {m \choose n_x} {m-n_x \choose n_z}
      2^{m-n_x-n_z} p_x^{m+n_x-n_z} p_z^{m-n_x+n_z} \,,
  \end{align}
  where the last equality follows form \Cref{prop:margin}. This shows equation
  \eqref{eq:result-abort}. 

  We proceed with showing equation \eqref{eq:result-sample}. We have that
  \begin{align}
    \Forall \basis \in \lk: \quad P_\Basis(\basis) 
    &= \sum_{(u,v) \in \Basis^{-1}(\basis)} P_{UV}(u,v) \\
    &= \sum_{(u,v) \in h^{-1}(\basis)} P_{UV}(u,v) \,,
  \end{align}
  where
  \begin{align}
    \label{eq:theta-preimage}
    h^{-1}(\basis) = \left\{ (u,v) \in \mn \times \mk \ \middle\vert \
    \begin{array}{l}
      \abs{u \wedge v} = 0 \,, \\
      \basis_i = 0 \implies u_{\alpha_i(u,v)} = 1 \,, \\
      \basis_i = 1 \implies v_{\alpha_i(u,v)} = 1
    \end{array}
    \right\} \,.
  \end{align}
  Since $\abs{u \wedge v} = 0$ for all $(u,v) \in h^{-1}(\basis)$, we know from 
  \Cref{prop:margin} that 
  \begin{align}
    &\Forall (u,v) \in (u,v) \in h^{-1}(\basis): \nonumber \\
    &P_{UV}(u,v) = \sum_{n_x=n}^{m-k} \sum_{n_z=k}^{m-n_x} {m-l \choose n_x-n}
    {m-k-n_x \choose n_z-k} 2^{m-n_x-n_z} p_x^{m+n_x-n_z} p_z^{m-n_x+n_z} {n_x
    \choose n}^{-1} {n_z \choose k}^{-1} \,.
  \end{align}
  Thus,
  \begin{align}
    &\Forall \basis \in \lk: \nonumber \\
    &P_\Basis(\basis) = \abs{h^{-1}(\basis)} \sum_{n_x=n}^{m-k}
    \sum_{n_z=k}^{m-n_x} {m-l \choose n_x-n} {m-k-n_x \choose n_z-k}
    2^{m-n_x-n_z} p_x^{m+n_x-n_z} p_z^{m-n_x+n_z} {n_x
    \choose n}^{-1} {n_z \choose k}^{-1} \,. \label{eq:pbasistheta}
  \end{align}
  For every $\basis \in \lk$, the set $h^{-1}(\basis)$ is the set of all pairs
  $(u,v) \in \mn \times \mk$ such that the following two properties are
  satisfied: 
  \begin{itemize}
    \item $\abs{u \wedge v} = 0$,
    \item for the set $\alpha(u,v)$ as in equation \eqref{eq:alpha}, it holds
      for every $i \in \upto{m}$ that $u_{\alpha_i(u,v)} = 1$ if $\basis_i = 0$
      and $v_{\alpha_i(u,v)} = 1$ if $\basis_i = 1$.
  \end{itemize}
  Now note that the only thing that matters is the question which $l = n+k$
  elements of $\upto{m}$ form the subset $\alpha_i(u,v) \subseteq~\upto{m}$: for
  every subset $\alpha \subseteq \upto{m}$ of size $l$, there is exactly one
  pair $(u,v)$ which satisfies the above two properties such that $\alpha =
  \alpha_i(u,v)$. Hence, counting the elements of $h^{-1}(\basis)$ is the same
  as counting the $l$-element subsets of $\upto{m}$, and therefore
  \begin{align}
    \abs{h^{-1}(\basis)} = {m \choose n+k} \,.
  \end{align}
  This reduces equation \eqref{eq:pbasistheta} to
  \begin{align}
    &\Forall \basis \in \lk: P_\Basis(\basis) = {m \choose n+k}
    \sum_{n_x=n}^{m-k} \sum_{n_z=k}^{m-n_x} {m-n-k \choose n_x-n} {m-k-n_x
    \choose n_z-k} 2^{m-n_x-n_z} p_x^{m+n_x-n_z} p_z^{m-n_x+n_z} {n_x \choose
    n}^{-1} {n_z \choose k}^{-1} \,,
  \end{align}
  which is what we wanted to show.
\end{proof}

\newpage

\section{Efficiency calculation}
\label{app:efficiency}

Here we compare the efficiencies of iterative sifting and LCA
sifting. Recall from Equation \eqref{eq:efficiency} that we define the
efficiency $\eta$ of a sifting protocol as 
\begin{align}
  \eta = \frac{R}{M} \,,
\end{align}
where $R$ is the random variable of the number of rounds that are kept after
sifting and $M$ is the random variable of the total number of rounds performed
in the loop phase of the protocol. The efficiency $\eta$ depends on the
particular history of the protocol: different runs of the protocol may have
different efficiencies. Therefore, $\eta$ is a random variable. In the
following, $R_I$ and $M_I$ denote the random variables $R$ and $M$ for the
iterative sifting protocol, and $R_L$ and $M_L$ denote the corresponding random
variables for the LCA protocol.  Whereas in the case of iterative
sifting, the number $R_I$ is fixed and the number $M_I$ is a random variable,
the opposite is true for the LCA sifting protocol, where $M_L = m$ is
fixed but but $R_L$ is a random variable. (Note that the LCA sifting
protocol may abort, in which case $R_L=0$).

To compare the efficiencies of the two protocols, we calculate the expected
value of $\eta$ in each case. We first do this for the case of iterative
sifting. Recall that $A_r$, $B_r$ is the random variable of Alice's and Bob's
basis choice in round $r$, respectively, and that $N_d$ is the number of basis
disagreements. Then we have:
\begin{align}
  \expect{\eta_I} 
  &= \expect{\frac{R_I}{M_I}} \\
  &= (n+k) \expect{\frac{1}{M_I}} \\
  &= (n+k) \sum_{m=n+k}^\infty \frac{1}{m} P_{M_I}(m) \\
  &= (n+k) \sum_{m=n+k}^\infty \frac{1}{m} \sum_{n_d=0}^{m-n-k}
    P_{M_IN_d}(m,n_d) \\
  &= (n+k) \sum_{m=n+k}^\infty \frac{1}{m} \sum_{n_d=0}^{m-n-k} \left(
    P_{M_IN_dA_mB_m}(m,n_d,0,0) + P_{M_IN_dA_mB_m}(m,n_d,1,1) \right) \\
  &= (n+k) \sum_{m=n+k}^\infty \frac{1}{m} \sum_{n_d=0}^{m-n-k} \Bigg( (p_x^2)^n
    (p_z^2)^{m-n-n_d} (2p_xp_z)^{n_d} {m-1 \choose n_d} {m-n_d-1 \choose n-1} +
    \\ 
  &\hspace{4.5cm} (p_x^2)^{m-k-n_d} (p_z^2)^k (2p_xp_z)^{n_d} {m-1 \choose n_d}
    {m-n_d-1 \choose k-1} \Bigg) \nonumber \\
  &= (n+k) \sum_{m=n+k}^\infty \frac{1}{m} \sum_{n_d=0}^{m-n-k}
    (2p_xp_z)^{n_d} {m-1 \choose n_d} \Bigg( (p_x^2)^n (p_z^2)^{m-n-n_d}
    {m-n_d-1 \choose n-1} + \label{eq:etaIFinal} \\
  &\hspace{7.5cm} (p_x^2)^{m-k-n_d} (p_z^2)^k {m-n_d-1 \choose k-1} \Bigg) \,.
    \nonumber
\end{align}
For the case of the LCA sifting protocol, we have:
\begin{align}
  \expect{\eta_L} 
  &= \frac{R_L}{M_L} \\
  &= \frac{1}{m} \expect{R_L} \\
  &= \frac{1}{m} (n+k) P[N_x \geq n \land N_z \geq k] \\
  &= \frac{1}{m} (n+k) \sum_{n_d=0}^{m-n-k} P[N_x \geq n \land N_z \geq k \land
    N_d = d] \\
  &= \frac{n+k}{m} \sum_{n_d=0}^{m-n-k} \sum_{n_z=k}^{m-n_d-n} P[N_x \geq n
  \land N_z = n_z \land N_d = n_d] \\
  &= \frac{n+k}{m} \sum_{n_d=0}^{m-n-k} \sum_{n_z=k}^{m-n_d-n}
    (p_x^2)^{m-n_z-n_d} (p_z^2)^{n_z} (2p_xp_z)^{n_d} {m \choose n_d} {m-n_d
    \choose n_z} \,. \label{eq:etaFFinal}
\end{align}
The calculation of the expected efficiencies \eqref{eq:etaIFinal} and
\eqref{eq:etaFFinal} requires a lot of computational power. We wrote programs
that compute numerical lower bounds on $\expect{\eta_I}$ and $\expect{\eta_L}$
for the case where the probabilities are symmetric ($p_x = p_z = 1/2$) and where
the quotas coincide ($n=k$). A plot of these lower bounds is shown in
\Cref{fig:efficiency}. In order to plot the lower bound on $\expect{\eta_L}$, a
choice for $m$ had to be made for each value of $n=k$. Our program choses an $m$
which is likely to maximize the expected efficiency for the given value of
$n=k$. Note that 1/2, being the expected fraction of basis agreements, is an
upper bound on the expected efficiencies. Hence, \Cref{fig:efficiency} indicates
that the difference in the expected efficiencies becomes insignificant for
practically relevant values of the block length $n+k$. This means that replacing
iterative sifting by LCA sifting is unlikely to have a significant effect on the
key rate of a QKD protocol.

\newpage

\section{Proof of the sufficiency of the formal criteria}
\label{app:formal-criteria}

In this appendix, we prove that the two formal criteria for good sifting,
\eqref{eq:uniform-formal} and \eqref{eq:uncorr-formal}, are sufficient for good
sifting in the sense that the relevant statistical inequality,
\eqref{eq:ptail-bound2}, follows from these two conditions. In other words, we
prove \Cref{prop:tail}. 

\begin{proof}[Proof of \Cref{prop:tail}]
  According to Bayes' Theorem, we have that
  \begin{align}
    \ptail &= P[\Lambdakey \geq \Lambdatest + \mu \mid \Lambdatest \leq \qtol]
      \\
    &= \frac{P[\Lambdatest \leq \qtol \mid \Lambdakey \geq \Lambdatest +
      \mu]P[\Lambdakey \geq \Lambdatest + \mu]}{P[\Lambdatest \leq \qtol]}
      \label{eq:ptail-bayes-1} \\
    &\leq \frac{P[\Lambdakey \geq \Lambdatest + \mu]}{\ppass} \,.
    \label{eq:ptail-bayes-2}
  \end{align}
  We define the \emph{total error rate} $\Lambdatot$ as the random variable
  \begin{align}
    \label{eq:ltot}
    \begin{array}{lccl}
      \Lambdatot: & \Omega_{ZZ'\Basis} & \rightarrow & [0,1] \\
      & (z,z',\basis) & \mapsto & {\displaystyle \frac{1}{l} \sum_{i=1}^l z
      \oplus z' } \,.  
    \end{array}
  \end{align}
  For all $(z,z',\basis) \in \Omega_{ZZ'\Basis}$, it holds that
  \begin{align}
    \Lambdakey(z,z,\basis) &\geq \Lambdatest(z,z,\basis) + \mu \\
    \iff \frac{1}{n} \sum_{i=1}^l (1-\basis_i)(z_i \oplus z'_i) &\geq
      \frac{1}{k} \sum_{i=1}^l \basis_i (z_i \oplus z'_i) + \mu \\
    \iff \frac{1}{n} \sum_{i=1}^l (1-\basis_i) (z_i \oplus z'_i) +
      \frac{1}{k} \sum_{i=1}^l (1-\basis_i) (z_i \oplus z'_i) &\geq \frac{1}{k}
      \sum_{i=1}^l \basis_i (z_i \oplus z'_i) + \frac{1}{k} \sum_{i=1}^l
      (1-\basis_i) (z_i \oplus z'_i) + \mu \\
    \iff \left( \frac{1}{n} + \frac{1}{k} \right) \sum_{i=1}^l (1-\basis_i) (z_i
      \oplus z'_i) &\geq \frac{1}{k} \sum_{i=1}^l (z_i \oplus z'_i) + \mu \\
    \iff \frac{k}{l} \left( \frac{1}{n} + \frac{1}{k} \right) \sum_{i=1}^l
      (1-\basis_i) (z_i \oplus z'_i) &\geq \frac{k}{l} \frac{1}{k} \sum_{i=1}^l
      (z_i \oplus z'_i) + \frac{k}{l} \mu \\
    \iff \Lambdakey(z,z,\basis) &\geq \Lambdatot(z,z,\basis) + \frac{k}{l} \mu
      \,.
    \label{eq:event-equiv}
  \end{align}
  We express the error \emph{rates} $\Lambdakey$, $\Lambdatest$ and $\Lambdatot$
  in terms of the error \emph{numbers}
  $\Sigmakey$, $\Sigmatest$ and $\Sigmatot$,
  \begin{align}
    \Sigmakey = n\Lambdakey \,, \quad \Sigmatest = k\Lambdatest \,, \quad
    \Sigmatot = l\Lambdatot \,.
  \end{align}
  This gives us
  \begin{align}
    \Lambdakey \geq \Lambdatot + \frac{k}{l}\mu \iff \Sigmakey \geq n \left(
    \frac{\Sigmatot}{l} + \frac{l-n}{l}\mu \right)
  \end{align}
  Therefore,
  \begin{align}
    \textstyle P[\Lambdakey \geq \Lambdatest + \mu] = P\left[ \Sigmakey \geq n
    \left( \frac{\Sigmatot}{l} + \frac{l-n}{l}\mu \right)\right]
  \end{align}
  and hence
  \begin{align}
    \label{eq:ptail-ltot}
    \ptail 
    &\leq \frac{{\textstyle P\left[ \Sigmakey \geq n \left(
      \frac{\Sigmatot}{l} + \frac{l-n}{l}\mu \right)\right]}}{\ppass} \\
    &= \frac{ \sum_{\sigmatot} P[\Sigmatot = \sigmatot] P\left[\Sigmakey \geq
      n \left( \frac{\sigmatot}{l} + \frac{l-n}{l} \mu \right) \Mid
      \Sigmatot = \sigmatot \right] }{\ppass} \\
    &= \frac{ \sum_{\sigmatot} P[\Sigmatot = \sigmatot] \sum_{j}
      P\left[\Sigmakey = j \Mid \Sigmatot = \sigmatot \right] }{\ppass} \,,
      \label{eq:sigmatot-split}
  \end{align}
  where the sum over $j$ ranges over all possible values of $\Sigmakey$ that are
  larger or equal to the according value, i.e.
  \begin{align}
    j = \left\lceil n \left( \frac{\sigmatot}{l} + \frac{l-n}{l} \mu \right)
    \right\rceil, \left\lceil n \left( \frac{\sigmatot}{l} + \frac{l-n}{l} \mu
    \right) \right\rceil +1, \ldots, n \,,  
    \label{eq:j-range}
  \end{align}
  where $\lceil\,\cdot\,\rceil$ denotes the ceiling function.
  \begin{align}
     h(\sigmatot, l, n, j) := \ &P\left[\Sigmakey = j \Mid \Sigmatot =
       \sigmatot \right] \\
     = \ &\frac{P[\Sigmakey=j \land
       \Sigmatot=\sigmatot]}{P[\Sigmatot=\sigmatot]} \\
     = \ &\frac{P[\Omega_{j\sigmatot}]}{P[\Omega_{\sigmatot}]} \,, \\
     = \ &\frac{\sum_{(z,z',\basis)\in\Omega_{j\sigmatot}}
     P_{ZZ'\Basis}(z,z',\basis)}{\sum_{(z,z',\basis)\in\Omega_{\sigmatot}}
       P_{ZZ'\Basis}(z,z',\basis)} \label{eq:h-def}
  \end{align}
  where
  \begin{align}
    \Omega_{j\sigmatot} &= \left\{ (z,z',\basis) \in \Omega_{ZZ'\Basis} \mid
      \Sigmakey(z,z',\basis)=j \land \Sigmatot(z,z',\basis)=\sigmatot \right\}
      \,, \\
    \Omega_{\sigmatot} &= \left\{ (z,z',\basis) \in \Omega_{ZZ'\Basis} \mid
      \Sigmatot(z,z',\basis)=\sigmatot \right\} \,.
  \end{align}
  It holds for all $(z,z',\basis) \in \Omega_{ZZ'\Basis}$ that
  \begin{align}
    P_{ZZ'\Basis}(z,z',\basis) 
    &= P_{ZZ'}(z,z') P_{\Basis}(\basis) \label{eq:factors} \\
    &= P_{ZZ'}(z,z') c \,, \label{eq:unifs}
  \end{align}
  where $P_{ZZ'}$ and $P_{\Basis}$ are the according marginal distributions of
  $P_{ZZ'\Basis}$. Equation \eqref{eq:factors} follows from
  \eqref{eq:uncorr-formal}, and \Cref{eq:unifs} follows from
  \Cref{eq:uniform-formal}.
  This implies
  \begin{align}
    h(\sigmatot,l,n,j) 
    &= \frac{\sum_{(z,z',\basis)\in\Omega_{j\sigmatot}}P_{ZZ'}(z,z')p}{
      \sum_{(z,z',\basis)\in\Omega_{\sigmatot}}P_{ZZ'}(z,z')p}
      \\ 
    &= \frac{\sum_{(z,z',\basis)\in\Omega_{j\sigmatot}}
      P_{ZZ'}(z,z')}{\sum_{(z,z',\basis)\in\Omega_{\sigmatot}}P_{ZZ'}(z,z')} \\
    &= \frac{\sum_{(z,z')\in\Gamma_{\sigmatot}}P_{ZZ'}(z,z') {\sigmatot \choose
      j} {l-\sigmatot \choose n-j} }{\sum_{(z,z') \in
      \Gamma_{\sigmatot}}P_{ZZ'}(z,z') {l \choose n} } \\
    &= {\sigmatot \choose j} {l-\sigmatot \choose n-j} {l \choose n}^{-1} \,,
      \label{eq:h-is-hyper}
  \end{align}
  where
  \begin{align}
    \Gamma_{\sigmatot} = \left\{ (z,z') \in \lcube \times \lcube \Mid \sum_{i=1}^l z_i
    \oplus z'_i = \sigmatot \right\} \,.
  \end{align}
  \Cref{eq:h-is-hyper} means that $h(\sigmatot,l,n,j)$ is a hypergeometric
  distribution. We are interested in the \emph{tail} of this distribution,
  \begin{align}
    H(\sigmatot,l,n,d) := \sum_{j=d}^n h(\sigmatot,l,n,j) \,,
  \end{align}
  because according to \Cref{eq:sigmatot-split,eq:j-range},
  \begin{align}
    \ptail \leq \frac{\sum_{\sigmatot}P[\Sigmatot=\sigmatot]
    H(\sigmatot,l,n,d)}
    {\ppass} \,,
    \label{eq:ptail-is-hypertail}
  \end{align}
  where
  \begin{align}
    d = \left\lceil n \left( \frac{\sigmatot}{l} + \frac{l-n}{l} \mu \right)
    \right\rceil \,.
  \end{align}
  There are several well-known bounds on the tail of a hypergeometric
  distribution \cite{BLM13}. For our case, \emph{Serfling's bound} is a suitable
  one \cite{Ser74}. The appropriate special case of Serfling's bound for this
  case reads
  \begin{align}
    H(\sigmatot,l,n,d) 
    &\leq \exp\left( -2\frac{(l-n)n}{l}\frac{l-n}{l-n+1}\mu^2 \right) \\
    &= \exp\left( -2\frac{kn}{l}\frac{k}{k+1}\mu^2 \right) \,.
    \label{eq:serfling}
  \end{align}
  (Instead of Serfling's bound, one may use \emph{Hoeffding's bound}
  \cite{Hoe63}.  That bound is weaker than Serfling's bound in this case, but it
  has the advantage that it has been formulated directly in terms of
  hypergeometric distributions \cite{Chv78,Ska13}, so these references are
  easier to understand in our context.) Inequalities
  \eqref{eq:ptail-is-hypertail} and \eqref{eq:serfling} together imply
  \begin{align}
    \ptail 
    &\leq \frac{\sum_{\sigmatot}P[\Sigmatot=\sigmatot]
    H(\sigmatot,l,n,d)}
      {\ppass} \\
    &\leq \frac{\exp\left( -2\frac{kn}{l}\frac{k}{k+1}\mu^2 \right)}{\ppass} \,,
  \end{align}
  which completes the proof.
\end{proof}

\end{document}